\documentclass[11pt]{article}

\usepackage{geometry}
\usepackage{graphicx}
\usepackage{amsmath, amssymb, enumerate, color, colortbl, framed, float}
\usepackage{amsfonts, amsthm, comment, longtable, caption}
\usepackage{thmtools, thm-restate, subfig}
\usepackage{natbib}
\allowdisplaybreaks
\usepackage{setspace}
\newcommand{\ds}{\displaystyle}
\newcommand{\E}{\text{E}}
\newcommand{\Var}{\text{Var}}
\newcommand{\Cov}{\text{Cov}}
\newcommand{\X}{\mathcal{X}}
\newcommand{\B}{\mathcal{B}}
\newcommand{\real}{\mathbb{R}}
\newcommand\numberthis{\addtocounter{equation}{1}\tag{\theequation}}
\renewcommand{\arraystretch}{1.5}
\allowdisplaybreaks

\usepackage[utf8]{inputenc} 
\usepackage[T1]{fontenc}    
\usepackage{hyperref}       
\usepackage{url}            
\usepackage{lipsum}

\newcommand{\RN}[1]{%
  \textup{\uppercase\expandafter{\romannumeral#1}}%
}
\newtheorem{theorem}{Theorem}
\newtheorem{result}{Result}

\newtheorem{lemma}{Lemma}
\newtheorem{ass}{Assumption}

\theoremstyle{remark}
\newtheorem{remark}{Remark}
\definecolor{Gray}{gray}{0.9}

\title{Estimating Monte Carlo variance from multiple Markov chains}

\author{
  Kushagra Gupta \\
  Department of Statistics\\
  Stanford University\\
  \texttt{kushgpt@stanford.edu} \and \
 Dootika Vats\footnote{corresponding author} \\
  Department of Mathematics and Statistics\\
  IIT Kanpur\\
  \texttt{dootika@iitk.ac.in} \\
}

\begin{document}
\maketitle
\onehalfspacing
\begin{abstract}
 Modern computational advances have enabled easy parallel implementations of Markov chain Monte Carlo (MCMC). However, almost all work in estimating the variance of Monte Carlo averages, including the efficient batch means (BM) estimator, focuses on a single-chain MCMC run. We demonstrate that simply averaging covariance matrix estimators from multiple chains can yield critical underestimates in small Monte Carlo sample sizes, especially for slow-mixing Markov chains.  We extend the work of \cite{arg:and:2006} and propose a multivariate replicated batch means (RBM) estimator that utilizes information from parallel chains, thereby correcting for the underestimation. Under weak conditions on the mixing rate of the process, RBM is strongly consistent and exhibits similar large-sample bias and variance to the BM estimator. We also exhibit superior theoretical properties of RBM by showing that the (negative) bias in the RBM estimator is less than the average BM estimator in the presence of positive correlation in MCMC. Consequently, in small runs, the RBM estimator can be dramatically superior and this is demonstrated through a variety of examples.

\end{abstract}

\section{Introduction}

Markov chain Monte Carlo (MCMC) algorithms have emerged as an essential tool for Bayesian inference when obtaining independent samples from a posterior is either inefficient or impossible. With increased computational power and the presence of multiple cores on personal computers, running parallel Markov chains dispersed over the state space is common for improved inference. On the other hand, output analysis in MCMC has predominantly focused on single-chain MCMC output with the understanding that multiple-chain analysis can follow from a simple averaging of information. We demonstrate that a simple averaging of the Monte Carlo limiting covariance matrix can yield undesirable results, warranting specific methods for multiple-chain output analysis. 

Let $F$ be the target distribution defined on a $d$-dimensional space $\X$, equipped with a countably generated $\sigma$-field. Let function $h:\X \to \mathbb{R}^p$ be such that $\mu = \int_{\X} h(x) F(dx)$ is a quantity of interest and $m$ independent $F$-ergodic Markov chains $\{X_{1,t}\}_{t\geq 1}, \dots, \{X_{m,t}\}_{t\geq 1}$ are simulated for this task. The Monte Carlo estimator of $\mu$ from the $k$th Markov chain satisfies
\[
\hat{\mu}_{k} := \dfrac{1}{n} \ds \sum_{t=1}^{n} h\left(X_{k,t} \right) \overset{a.s.}{\to} \mu \text{ as } n \to \infty\,,
\]
where the convergence holds for any starting value $X_{k,1}$. Consequently, the combined estimator of $\mu$ from the $m$ independent chains is $\hat{\mu} = m^{-1} \sum_{k=1}^{m} \hat{\mu}_k\,.$
If a Markov chain central limit theorem (CLT) holds, then irrespective of the starting value, there exists a $p \times p$ positive-definite matrix $\Sigma$ such that for all $k = 1, \dots, m$,  as $n \to \infty$,
\begin{equation*}
\label{eq:multi_CLT_each}
\sqrt{n}(\hat{\mu}_k - \mu) \overset{d}{\to} N_p(0, \Sigma)\,.
\end{equation*}
Consequently, a CLT holds for $\hat{\mu}$ as well so that as $n \to \infty$,
\begin{equation}
\label{eq:multi_CLT}
\sqrt{n}(\hat{\mu} - \mu) \overset{d}{\to} N_p \left(0, \dfrac{\Sigma}{m} \right)\,.	
\end{equation}
A significant goal in MCMC output analysis is estimating $\Sigma$ in order to ascertain the quality of estimation of $\mu$ via $\hat{\mu}$. This allows for determining when to stop the MCMC sampler \citep{bro:gel:1998,fleg:hara:jone:2008,gelm:rubi:1992a,gong:fleg:2016,roy:2019,vats:fleg:jon:2019,vehtari2020rank}. Additionally, when used with sequential stopping rules, estimators of $\Sigma$ must be (strongly) consistent in order to yield asymptotic nominal coverage \citep{glyn:whit:1992}. Amongst the list of estimators, only the estimator of \cite{broo:robe:1998}, which estimates $\Sigma$ by the sample covariance matrix of $\hat{\mu}_1, \dots, \hat{\mu}_m$,  is specifically constructed for multiple-chains. We show that this estimator is not consistent in $n$ and as expected exhibits high variability when the number of parallel implementations is small. 

Many estimators have been introduced for a single-chain implementation, including batch means (BM) estimators \citep{chen:seila:1987}, spectral variance estimators \citep{andr:1991,hannan:1970}, regenerative estimators \citep{seil:1982}, and initial sequence estimators \citep{berg2022efficient,dai:jon:2017,koso:2000}. All of the estimators have their pros and cons, and \cite{vats:rob:fle:jon:2020} recommend the BM estimator due to its ease of implementation, universality of application, computational efficiency, and well established asymptotic properties.

Let $\hat{\Sigma}_k$ denote the BM estimator of $\Sigma$ from the $k$th chain; the estimator is detailed in Section~\ref{sec:BM_estimators}. A default combined estimator of $\Sigma$, which we call the average batch means (ABM) estimator is, $\hat{\Sigma} = m^{-1} \sum_{k=1}^{m} \hat{\Sigma}_{k}\,.$
ABM naturally retains asymptotic properties of $\hat{\Sigma}_k$, since the $m$ chains are independent. This averaging of the variance is reasonable for independent and identically distributed samples. However, even under stationarity, before each chain has adequately explored the state space, chains may dominate different parts of $\X$, so that each $\hat{\Sigma}_{k}$ significantly underestimates the truth. As a result, the ABM estimator also underestimates the truth and many of the benefits of running parallel chains from dispersed starting values are lost.

In the steady-state literature for univariate $\phi$-mixing processes ($p = 1$), \cite{arg:and:2006} introduced a \textit{replicated batch means} (RBM) estimator that estimates $\Sigma$ by essentially centering the chain around the overall mean $\hat{\mu}$, instead of the $k$th chain mean, $\hat{\mu}_k$. In most practical MCMC applications, $h$ is multivariate, and \cite{arg:and:2006} does not apply. Further, only a handful of Markov chains are known to be $\phi$-mixing. Thus, the RBM estimator presented in \cite{arg:and:2006} is not immediately applicable to MCMC. We develop a \textit{multivariate RBM} estimator and obtain (strong) consistency and large-sample bias and variance results for the larger class of $\alpha$-mixing processes. Further, we also show that RBM exhibits a smaller bias in finite samples compared to ABM for Markov chain exhibiting positive correlation. 

Asymptotically, multivariate RBM and ABM are no different. However, since deviation is measured from the global mean, if the chains explore different parts of the state space, RBM yields dramatic improvements over ABM.  These finite-sample advantages over ABM are illustrated through three examples: a Gibbs sampler for a bivariate Gaussian target, a Metropolis-Hastings sampler for the tricky Rosenbrock distribution, and a Hamiltonian Monte Carlo sampler for a Bayesian neural network posterior.  
 Over all three examples, the results are consistent: when dealing with slow mixing Markov chains or small Monte Carlo sample sizes, RBM can far outperform ABM, yielding a significantly more reliable estimator of $\Sigma$. For fast-mixing chains, the RBM estimator is at least as good as the ABM estimator with no additional computation cost.

The rest of the paper is organized as follows. In Section~\ref{sec:BM_estimators}, we present the existing single-chain BM estimators, and then outline the proposed multivariate RBM estimator. We also highlight other multiple-chain estimators of $\Sigma$ that we will employ for comparisons with the RBM estimator. The theoretical properties of the RBM estimator are outlined in Section~\ref{subsec:main_results}, and implementation in a variety of examples is shown in Section~\ref{sec:examples}.

\section{Batch-means estimators}
\label{sec:BM_estimators}
Recall that $\{X_{k,t}\}_{t=1}^{n}$ denotes the $k$th Markov chain of Monte Carlo sample size $n$ and define $\{Y_{k,t}\}_{t=1}^{n} = \{ h(X_{k,t})\}_{t=1}^{n}$. Let $n = a_n b_n$, where both $a_n, b_n \in \mathbb{N}$, $b_n$ is the size of each batch, and $a_n$ denotes the number of batches.  Let $\bar{Y}_{k,l}$ denote the mean vector for the $l$th batch in the $k$th chain. That is,
\[
\bar{Y}_{k,l} = \dfrac{1}{b_n} \ds \sum_{t=1}^{b_n} Y_{k,(l-1)b_n + t} \qquad l = 1, \dots, a_n\,.
\]
%
For the $k$th chain, \cite{chen:seila:1987} define the multivariate BM estimator of $\Sigma$ as
\begin{equation}
\label{eq:singe_bm}
  \hat{\Sigma}_{k,b_n} = \dfrac{b_n}{a_n - 1}\sum_{l = 1}^{a_n}\left(\bar{Y}_{k,l} - \hat{\mu}_k\right)\left(\bar{Y}_{k,l} - \hat{\mu}_k\right)^{T} \,.
\end{equation}
BM estimators have been used extensively in steady-state simulation and MCMC as they are computationally adept at handling the typical high-dimensional output generated by simulations. In the context of MCMC, desirable asymptotic properties of BM estimators are well established; \cite{jone:hara:caff:neat:2006,vats:fleg:jon:2019} provide necessary conditions for strong consistency, \cite{fleg:jone:2010,vats:flegal:2018} show mean-square-consistency, and \cite{chak:khare:2019} demonstrate asymptotic normality of the univariate BM estimator.

Despite desirable theoretical properties, the traditional BM estimator is inclined to exhibit finite-sample negative bias \citep[see][for e.g.]{damerdji:1995}, particularly for slow mixing Markov chains. To overcome this, \cite{liu:fleg:2018,vats:flegal:2018} propose the lugsail BM estimator that uses a jackknife-like trick to offset the negative bias in the opposite direction. For $r \geq 1$ and $0 \leq c < 1$, the lugsail BM estimator for the $k$th chain is:
\[
\hat{\Sigma}_{k,L} = \dfrac{1}{1 - c}\hat{\Sigma}_{k,b_n} - \dfrac{c}{1 - c}\hat{\Sigma}_{k, \lfloor b_n/r \rfloor} \; ,
\]
where $\hat{\Sigma}_{k, \lfloor b_n/r \rfloor}$ is the BM estimator with batch size $\lfloor b_n/r \rfloor$. Setting $r = 1$ yields the usual BM estimator in \eqref{eq:singe_bm}. \cite{vats:flegal:2018} recommend $c = 1/2$ and $r = 3$ when the underlying Markov chain is slowly mixing. 

\subsection{Multivariate replicated batch-means estimator} 
\label{sub:replicated_batch_means_estimator}

The traditional batch-means estimator is built for a single-run Markov chain. As will be demonstrated in the examples, in the event of slow mixing Markov chains, a simple average of the BM estimators can yield inaccurate results.  \cite{arg:and:2006} tackled this problem for the univariate case by constructing replicated batch means estimator for $\phi$-mixing processes by essentially concatenating the $m$ Markov chains. We develop a more general, multivariate version of their estimator and establish desirable theoretical properties under realistic conditions on the Markov chain. For a batch size of $b_n$,  the multivariate RBM estimator is
\begin{equation}
\label{eq:RBM}
    \hat{\Sigma}_{R, b_n} = \dfrac{b_n}{a_n m-1}\sum_{k=1}^{m}\sum_{l=1}^{a_n}\left(\bar{Y}_{k,l} - \hat{\mu}\right)\left(\bar{Y}_{k,l} - \hat{\mu} \right)^T \; .
\end{equation}
The idea behind RBM is to pool all the $a_nm$ batch mean vectors and measure their variability from the overall mean, $\hat{\mu}$. Thus, if each Markov chain is in different parts of the state space due to slow mixing,  $\hat{\Sigma}_{R, b_n}$ acknowledges the high variability in the batch mean vectors. 
%
The regular RBM estimator inherits the negative bias of the standard single-chain BM estimator, particularly in the presence of high autocorrelation. Fortunately, a lugsail RBM estimator is easily implementable:
\[
\hat{\Sigma}_{R,L} = \dfrac{1}{1 - c} \hat{\Sigma}_{R,b_n} - \dfrac{c}{1 - c} \hat{\Sigma}_{R, \lfloor b_n/r \rfloor} \; ,
\]

In Section~\ref{subsec:main_results} we present our main results on the RBM estimator. Although, much of our focus is on MCMC, our results apply more generally to $\alpha$-mixing processes.  For a stationary stochastic process $\{S_{n}\}$ on a probability space $\left(\Omega, \mathbb{F}, P\right)$, set $\mathbb{F}_{s}^{l} = \sigma\left(S_{s},...,S_{l}\right)$. Define the $\alpha$-mixing coefficients for $n = 1,2,...$ as 
\[
\alpha(n) = \underset{s \geq 1}{\sup}\underset{A\in \mathbb{F}_{1}^{s}, B\in \mathbb{F}_{s+n}^{\infty}}{\sup} |P\left(A \cap B\right) - P\left(A\right)P\left(B\right) | \; .
\]
A process is $\alpha$-mixing (or strongly mixing) if $\alpha(n) \to 0$ as $n \to \infty$. Erogidc Markov chains employed in MCMC are $\alpha$-mixing by design, however their rate is of critical importance. Our first assumption is on the mixing rate of $\{X_{k,t}\}$ and the moments of $\{Y_{k,t}\}  = \{h(X_{k,t})\}$, for a function $h:\X \to \mathbb{R}^p$.
\begin{ass}
\label{ass:mixing}
For some $\delta > 0$ and $\epsilon > 0$ such that $\E_F\|Y_{k,1}\|^{2 + \delta + \epsilon} < \infty$ and $\{X_{k,t}\}$ is $\alpha$-mixing with $\alpha(n) = o\left(n^{-\left(2 + \delta\right)\left(1 + (2  +\delta)/\epsilon \right)}\right)$ for all $k = 1, \dots, m$.
\end{ass}

Assumption~\ref{ass:mixing} will be our umbrella assumption on the process, however, we bring special focus to when this assumption is satisfied by Markov chains.  
Let $K: \X \times \B(\X) \to [0,1]$ be an $F$-invariant Markov chain transition kernel, where for $x \in \X$ and $A \in \B(\X)$, $K(x, A) := \Pr(X_2 \in A |X_1 = x)$. The $n$-step transition is $K^n(x,A) = \Pr(X_{n+1} \in A \mid X_1 = x)$. Standard MCMC algorithms are constructed so that $K$ is $F$-Harris ergodic, which guarantees ergodicity and implies that, in addition to convergence to the target distribution $F$, the Markov chain is $\alpha$-mixing \citep{jone:2004}.  Assumption~\ref{ass:mixing} may hold when the Markov chain is \textit{polynomially ergodic}. That is, suppose there exists $M:\X \to [0, \infty)$ with $\E_F M < \infty$ and $\xi > 0$ such that
\[
\|K^n(x,\cdot) - F(\cdot) \|_{TV} \leq M(x) n^{-\xi}\,,
\]
where $\|\cdot\|_{TV}$ is the total variation norm. Polynomially ergodic Markov chains of order $\xi > \left(2 + \epsilon\right)\left(1 + (2 + \delta)/\epsilon \right)$ satisfy Assumption~\ref{ass:mixing} \citep{jone:2004}.


\subsection{Other multiple chain estimators} 
\label{sub:other_multiple_chain_estimators}

Despite the rich and thorough literature on the estimation of $\Sigma$, there has been little work done for the case when parallel chains are employed. 
The only known estimator that utilizes this feature is the \textit{naive} estimator used in the convergence diagnostic of \cite{bro:gel:1998}, which is obtained by taking the sample covariance of the $m$ sample means $\hat{\mu}_k$. That is,
\begin{equation}
\label{eq:naive}
  \hat{\Sigma}_{N} := \dfrac{n}{m-1}\sum_{k=1}^{m}\left(\hat{\mu}_k - \hat{\mu}\right)\left(\hat{\mu}_k - \hat{\mu}\right)^{T}\;.
\end{equation}
In principle, $\hat{\Sigma}_{N}$ would perform well if the number of parallel runs, $m$, was reasonably large. In fact, if the asymptotics are in both $m$ and $n$, $\hat{\Sigma}_N$ is a viable estimator. However, typically, the number of parallel implementations of a Markov chain is small, and the naive estimator has a high variance even for large $n$. The following theorem shows that $\hat{\Sigma}_N$ is not consistent for $\Sigma$ for fixed $m$; the proof is in Appendix \ref{subsec: naive_convergence}.
\begin{theorem}
\label{thm:naive_convergence}
  Under Assumption~\ref{ass:mixing}, $\hat{\Sigma}_N \overset{d}{\to} (m-1)^{-1} \text{Wishart}_p(\Sigma, m-1)$ as $n \to \infty$. 
\end{theorem}
 We will demonstrate the instability of $\hat{\Sigma}_N$ in Section~\ref{sec:examples}, mimicking what was witnessed by \cite{fleg:hara:jone:2008}.



Given the lack of literature on estimators of $\Sigma$ under the parallel chain regime, we discuss a possible default estimator. A natural way to combine variance estimators from $m$ chains is to average them. This yields a lugsail ABM estimator to be:
\begin{equation}
\label{eq:abm}
 \hat{\Sigma}_{A,L} = \dfrac{1}{m}\sum_{k=1}^{m}\hat{\Sigma}_{k,L}\,. 
\end{equation}
%
Since the $m$ Markov chains are independent, large sample properties of $\hat{\Sigma}_{k,L}$ are shared by $\hat{\Sigma}_{A,L}$. This includes strong consistency, which was discussed in \cite{vats:fleg:jon:2019}, a CLT analogue for the univariate estimator \citep{chak:khare:2019},  and large-sample bias and variance results \citep{fleg:jone:2010,vats:flegal:2018}. 


\subsection{Main results} 
\label{subsec:main_results}

In this section, we establish some critical results for $\hat{\Sigma}_{R}$ and $\hat{\Sigma}_{R,L}$. Particularly, interest is in ensuring that  estimators of $\Sigma$  (i) are strongly consistent to allow validity of sequential stopping rules \citep{glyn:whit:1992}, (ii) exhibit a tractable form of large-sample bias, (iii) exhibit manageable variance in estimation. As is common in the literature, we will require the following result that establishes the existence of a strong invariance principle under Assumption~\ref{ass:mixing}. Let ${B\left(t\right)}$ be a $p$-dimensional standard Brownian motion.
%
%
%


\begin{theorem}[\cite{banerjee2022multivariate}]
\label{thm:SIP}
Let Assumption~\ref{ass:mixing} hold. Then there exists a $p \times p$ lower triangular matrix $L$ such that $LL^T = \Sigma$, a finite random variable D, and a sufficiently rich probability space $\Omega$ such that for almost all $w \in \Omega$ and for all $n > n_{0}$, with probability 1,
\[
\left\|\sum_{t=1}^{n}Y_{k,t} - n \mu - LB\left(n\right)\right\| < D(w) \psi(n)\;,
\] 
where $\psi(n) = n^{\max\{1/(2 + \delta), 1/4\}} \log(n)$.
\end{theorem}

\begin{ass}
\label{ass:bn}
  The integer sequence $b_n$ is such that $b_n \to \infty$ and $b_n/n \to 0$ as $n \to \infty$, and both $b_n$ and $n/b_n$ are non decreasing.
\end{ass}
Assumption~\ref{ass:bn} is necessary for (weak) consistency and mean-square consistency of BM estimators \citep{glyn:whit:1991,damerdji:1995}. Additional assumptions on $a_n$ and $b_n$ may be needed for additional results. Our first main result, the proof for which is in Appendix~\ref{subsec: strong_consistency_RBM}, demonstrates strong consistency of the lugsail RBM estimator.  The following theorem establishes that if the single-chain BM estimators are strongly consistent, then the RBM estimator is strongly consistent. 
\begin{theorem}
\label{thm:RBM_consistency}
  If Assumptions~\ref{ass:mixing} and \ref{ass:bn} hold, $\hat{\Sigma}_{k,b_n} \overset{a.s.}{\to} \Sigma$, and $\left(b_n\log \log n\right)/n \to 0$, then $\hat{\Sigma}_{R,L}\overset{a.s.}{\to} \Sigma$ as $n \to \infty$.
\end{theorem}

The strong consistency result in Theorem~\ref{thm:RBM_consistency} guarantees that sequential stopping rules of the style of \cite{glyn:whit:1992} lead to asymptotically valid confidence regions. Specifically for MCMC, Theorem~\ref{thm:RBM_consistency} allows the usage of effective sample size as a valid termination rule. We will discuss the impact of this in Section~\ref{sec:examples}.

Our next set of results focus on the bias and variance of the RBM estimator. Obtaining finite-time results on the bias and variance of BM estimators is challenging due to the underlying correlation in the mixing processes. However, large-sample results are possible, and we study these here.  Let $\hat{\Sigma}^{ij}_{R,L}$ and $\Sigma_{ij}$ be the $ij$th element of $\hat{\Sigma}_{R,L}$ and $\Sigma$, respectively. Let $R_s = \Cov_F(Y_{1,1}, Y_{1, 1+s})$ denote the $s$-lag covariance matrix. Denote
\[
\Gamma = - \sum_{s=1}^{\infty}s\left[R_s + R_s^T\right]\,.
\]

The proofs of the following two theorems are in Appendix~\ref{subsec:bias_RBM} and Appendix~\ref{subsec:variance_RBM}.
\begin{theorem}
\label{thm:rbm_bias}

If Assumption~\ref{ass:mixing} holds,
\[
\text{Bias}\left(\hat{\Sigma}_{R,b_n} \right) = \left(\dfrac{1-rc}{1-c} \right)\dfrac{\Gamma}{b_n} + o\left(\dfrac{1}{b_n}\right) \; .
\]
Further, under Assumption~\ref{ass:bn}, $\hat{\Sigma}_{R,b_n}$ is asymptotically unbiased. 
\end{theorem}

\begin{theorem}
\label{thm:rbm _variance}
Using the notation from Theorem \ref{thm:SIP}, if Assumptions~\ref{ass:mixing} and \ref{ass:bn} hold, $\E D^{4} < \infty$, $E_{F}\|Y_{1,1}\|^{4} < \infty$,  and $b^{-1}_n\psi^{2}(n)\log n \to 0$ as $n \to \infty$, then
\[
Var\left(\hat{\Sigma}^{ij}_{R,L}\right) = \dfrac{b_n}{nm}\left[\left(\dfrac{1}{r} + \dfrac{r - 1}{r\left(1 - c\right)^{2}}\right)\left(\Sigma^{2}_{ij} + \Sigma_{ii}\Sigma_{jj}\right)\right] + o\left(\dfrac{b_n}{n}\right) \; .
\]
\end{theorem}
Theorems~\ref{thm:rbm_bias} and \ref{thm:rbm _variance} together yield mean-square-consistency of the RBM estimator. The bias and variance expressions are essentially similar to the single-chain BM results. Thus, for large $n$, there is essentially no difference between ABM and RBM; this is also verified in all of our examples. However, and most critically, for small $n$ and particularly when any of the $m$ Markov chains have not sufficiently explored the state space, the RBM estimator exhibits far superior finite sample properties. We show that the first-order bias in RBM is much smaller than the first-order bias in ABM, establishing the systemic superiority of RBM which is illustrated in all of our examples in Section~\ref{sec:examples}.

\begin{theorem}
\label{thm:compare_bias}
Under Assumption~\ref{ass:mixing}, for each $k = 1, \dots, p$, the difference in the bias of RBM and ABM estimators is
\[
\E\left[\hat{\Sigma}_{R,b_n}^{kk} - \hat{\Sigma}_{A,b_n}^{kk} \right]= \dfrac{2a_n(m-1)}{a_nm-1} \left[\sum_{s=1}^{b_n-1} \dfrac{s}{n}R_{s}^{kk} + \dfrac{1}{a_n-1}\sum_{s=b_n}^{n-1}\left(1 - \dfrac{s}{n}\right) R_{s}^{kk} \right]\,.
\]

\end{theorem}
\begin{remark}
Typically, an MCMC algorithm exhibits positive autocorrelation,  $R_{s}^{kk} > 0$  for all $s$. Theorem~\ref{thm:rbm_bias} concludes that when in this scenario, these estimators exhibit negative bias. Theorem~\ref{thm:compare_bias} further elucidates that on average, RBM yields larger estimates than ABM, particularly for slow mixing chains, thereby reducing the underestimation.  Notice that when $m = 1$, the difference in the bias is exactly zero; this is because the ABM and RBM estimators coincide. 
\end{remark}


\section{Examples}
\label{sec:examples}

\subsection{Simulation Setup} 
\label{sub:simulation_setup}


In each example\footnote{ \texttt{R} implementations for the examples are available here: \url{https://github.com/dvats/RBMPaperCode} }, we compare the performance of RBM against ABM and the naive estimator. We implement lugsail versions of ABM and RBM with $r = 3$ and $c = 1/2$, as recommended for MCMC by \cite{vats:flegal:2018}. For the first two applications, we start all Markov chains from stationarity by drawing the initial value from the target distribution. This ensures that there is no issue with a misdiagnosed burn-in and the numerical analysis is performed on stationary chains. For the third example where the true target is unknown, we remove the initial 1000 samples as that was around the expected time for the Markov chain to reach a local mode. 

For slow mixing Markov chains, we set the batch size to be  $b = \beta \lfloor n^{1/3} \rfloor$ where $\beta$ is estimated using the methods of \cite{fleg:jone:2010,liu2018batch}. For fast mixing Markov chains, we set $b = \lfloor n^{1/2} \rfloor$. For all examples, $h$ is the identity function. We assess the performance of the variance estimators by two metrics: (a) coverage probability (where appropriate) (b) running estimates of the determinant of the estimated covariance matrices, and  (c) quality of estimation of effective sample size (ESS). 

Using the CLT in \eqref{eq:multi_CLT}, we may construct large-sample ellipsoidal confidence regions around $\hat{\mu}$ using different variance estimators. These $(1-\alpha)100 \%$ regions are:
\[
\left \{\mu \in \real^p: n (\hat{\mu} - \mu)^T \hat{\Sigma}^{-1}(\hat{\mu} - \mu) \leq \chi^2_{1-\alpha,p} \right\}\,,
\]
where $\chi^2_{1-\alpha,p}$ is the $1-\alpha$ quantile of a $\chi^2$ distribution with $p$ degrees of freedom, and $\hat{\Sigma}$ is any of the three naive, ABM, or RBM. In the first two examples where the true $\mu$ is known, over repeated simulations we report the proportion of repetitions for which the constructed confidence regions contain the true $\mu$. Due to Theorem~\ref{thm:RBM_consistency}, RBM and ABM are both consistent and thus we expect this proportion to be close to $1-\alpha$. Since the naive estimator is inconsistent, we expect non-nominal coverage probabilities. In our first example, the true value of $\Sigma$ is known. Here, we also compare $\|\hat{\Sigma}\|_{\text{det}}$ with the true determinant of the covariance matrix.


The estimation of $\Sigma$ is most critical for the calculation of the effective sample size (ESS) \citep{gong:fleg:2016,vats:fleg:jon:2019,kass:carlin:gelman:neal:1998}. The ESS describes the number of independent samples that would yield the same quality of estimation as this sample. \cite{vats:knud:2018} establish a one-to-one relationship between the Gelman-Rubin-Brooks diagnostic \citep{gelm:rubi:1992a,broo:robe:1998} and ESS. They define ESS for multiple chains as
\[
\text{ESS} = mn \left(\dfrac{|\text{Var}_F(h(X_11))|}{|\Sigma|} \right)^{1/p}\,,
\]
where $|\cdot|$ denotes determinant. \cite{roy:2019,vats:knud:2018,vats:fleg:jon:2019} explain that MCMC simulations may be stopped with confidence when the estimated ESS is more than a pre-specified lower bound. Underestimating $\Sigma$ overestimates ESS leading to early termination of simulation and a false sense of security on the quality of estimation of $\mu$.


\subsection{Bivariate normal Gibbs sampler}
Consider sampling from a bivariate normal distribution using a Gibbs sampler. For $\omega_1, \omega_2 > 0$ and  $\rho$ such that $\rho^2 < \omega_1 \omega_2$, the target distribution is 
\[
\left(\begin{array}{c}
  X_1 \\ X_2
\end{array} \right) \sim 
 N \begin{pmatrix}
\begin{pmatrix}
\mu_{1}\\
\mu_{2}
\end{pmatrix}, 
\begin{pmatrix}
\omega_1 & \rho \\
\rho & \omega_2
\end{pmatrix}
\end{pmatrix}
\]
The Gibbs sampler updates the Markov chain by using the following full conditional distributions:
\begin{equation}
\label{eqn:gibbs1}
X_{1} \mid X_{2} \sim  N\left(\mu_{1} + \dfrac{\rho}{\omega_2}\left(X_{2} - \mu_{2}\right) , \omega_1 - \dfrac{\rho^{2}}{\omega_2}\right) \; .
\end{equation}
\begin{equation}
\label{eqn:gibbs2}
X_{2} \mid X_{1} \sim  N\left(\mu_{2} + \dfrac{\rho}{\omega_1}\left(X_{1} - \mu_{1}\right) , \omega_2 - \dfrac{\rho^{2}}{\omega_1}\right) \; .
\end{equation}
Although a seemingly simple example, the Gibbs sampler  can exhibit arbitrarily fast or slow mixing based on the correlation in target distribution. In fact, \cite{rob:sahu:1996} makes the case for approximating general Gibbs samplers with unimodal target distributions with a similar Gaussian Gibbs sampler. In the following theorem, we obtain the exact form of the asymptotic covariance matrix for estimating $(\mu_1, \mu_2)^T$ with the Monte Carlo average. Versions of the result are known to exist in the literature, and for completeness, we provide a proof in Appendix~\ref{subsec:Gibbs_var}.

\begin{result}
\label{thm:bivariate_gibbs}
  For the two variable deterministic scan Gibbs sampler described by the full conditionals in \eqref{eqn:gibbs1} and \eqref{eqn:gibbs2}, the asymptotic covariance matrix in the CLT for the sample mean is
\[
\Sigma = \begin{bmatrix}
\omega_1\left(\dfrac{\omega_1\omega_2 + \rho^{2}}{\omega_1 \omega_2 - \rho^{2}}\right) & \dfrac{2 \omega_1 \omega_2\rho}{\omega_1 \omega_2 - \rho^{2}}\\
\dfrac{2\omega_1 \omega_2\rho}{\omega_1\omega_2 - \rho^{2}} & \omega_2\left(\dfrac{\omega_1\omega_2 + \rho^{2}}{\omega_1\omega_2 - \rho^{2}}\right)
\end{bmatrix} \; .
\] 
\end{result}
Result~\ref{thm:bivariate_gibbs} allows us to compare our estimates of $\Sigma$ against the truth. We compare the performance of RBM, ABM, and the naive estimator for two settings, $\rho = 0.5, 0.999$ with $\omega_1 = \omega_2 = 1$. For Setting 1 ($\rho = .5$), the Markov chains travel freely across the state space so that in only a few steps, independent runs of the Markov chains look similar to each other. For Setting 2 $(\rho = .999)$, each Markov chain moves slowly across the state space. Trace plots in Figure~\ref{fig:MVG_trace} give evidence of this.
\begin{figure}[htbp]
    \centering
    \subfloat[][$\rho = 0.5$]{
         \centering
         \includegraphics[width= 2.4in]{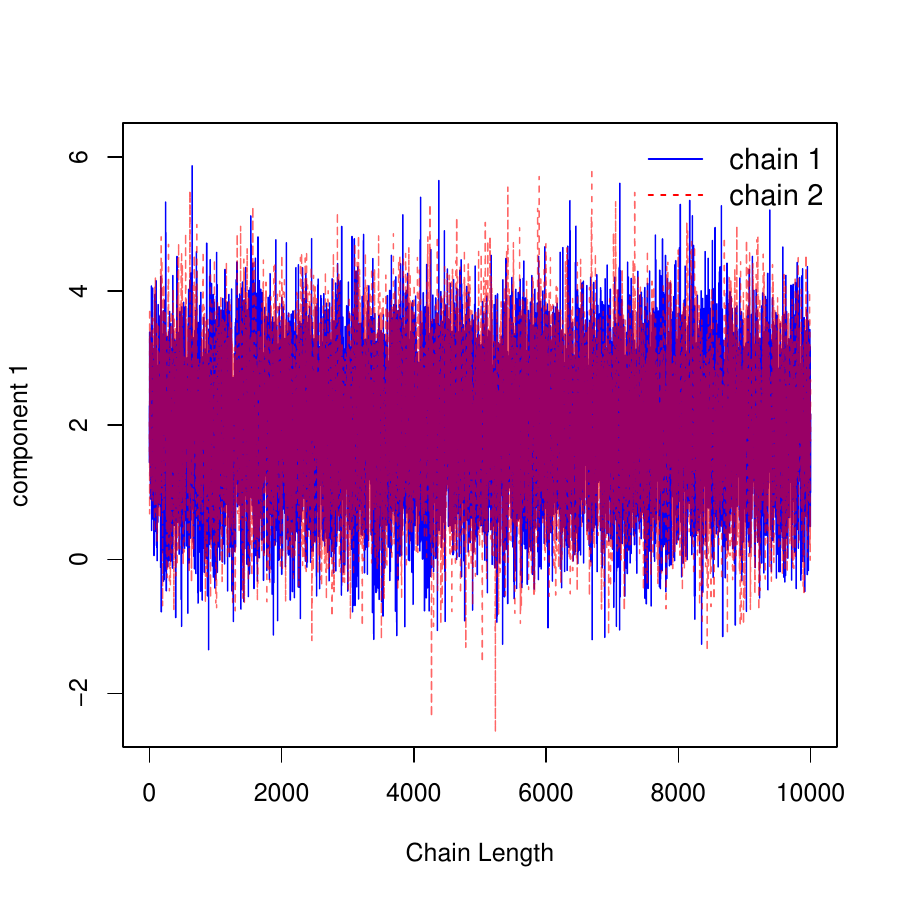} 
         \label{subfig:MVG_trace_rho_0.5}
}
    \subfloat[][$\rho = 0.999$]{
         \centering
         \includegraphics[width= 2.4in]{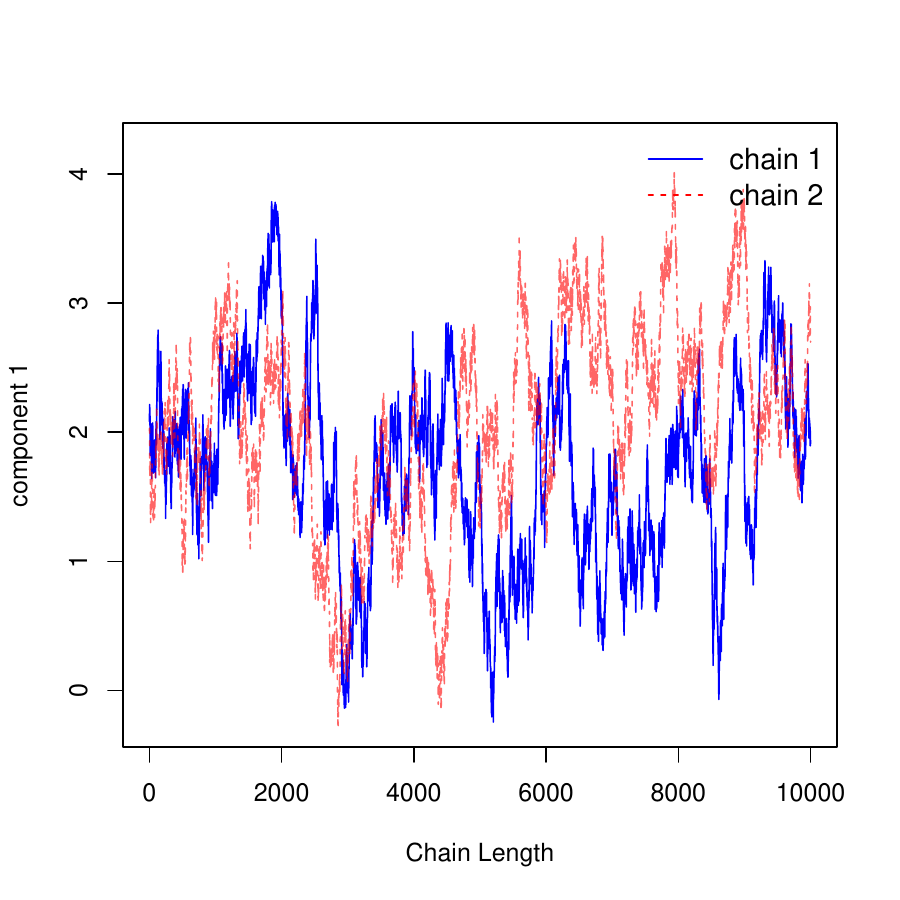}
         \label{subfig:MVG_trace_rho_0.999}
}
    \caption{Bivariate normal: Trace plots of two parallel Markov chains.}
    \label{fig:MVG_trace}
\end{figure} 

We set $m = 5,10$ for both the settings and first compare the evolution of the estimates of $\Sigma$ over time. 
 Figure~\ref{fig:MVG_FN_running} presents the running plots of the determinant of the estimators with the black horizontal line being the truth. In all four plots, the high variability of the naive estimator is demonstrated, consistent with Theorem~\ref{thm:naive_convergence}. For Setting 1, RBM and ABM are indistinguishable as was expected from the trace plots. However, the advantage of using RBM is evident in Setting 2, where the RBM estimator converges to the truth slightly quicker than the ABM estimator. This is a direct consequence of the fact that slowly mixing Markov chains take longer to traverse the state space, implying that the means of all the chains are different in small MCMC runs.

\begin{figure}[htbp]
    \centering
    \subfloat[][$m = 5, \rho = 0.5$]{
         \centering
         \includegraphics[width = 2.9in]{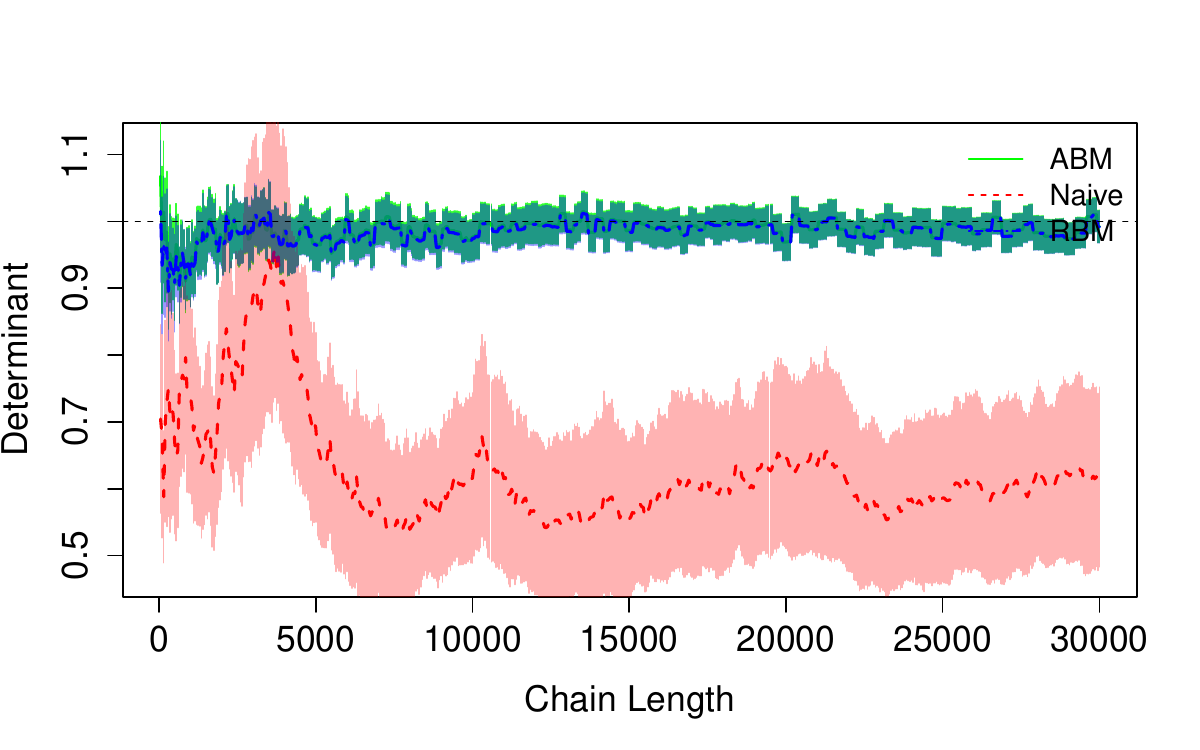}
         \label{subfig:MVG_running_rho_0.5}
}
    \subfloat[][$m = 5, \rho = 0.999$]{
         \centering
         \includegraphics[width = 2.9in]{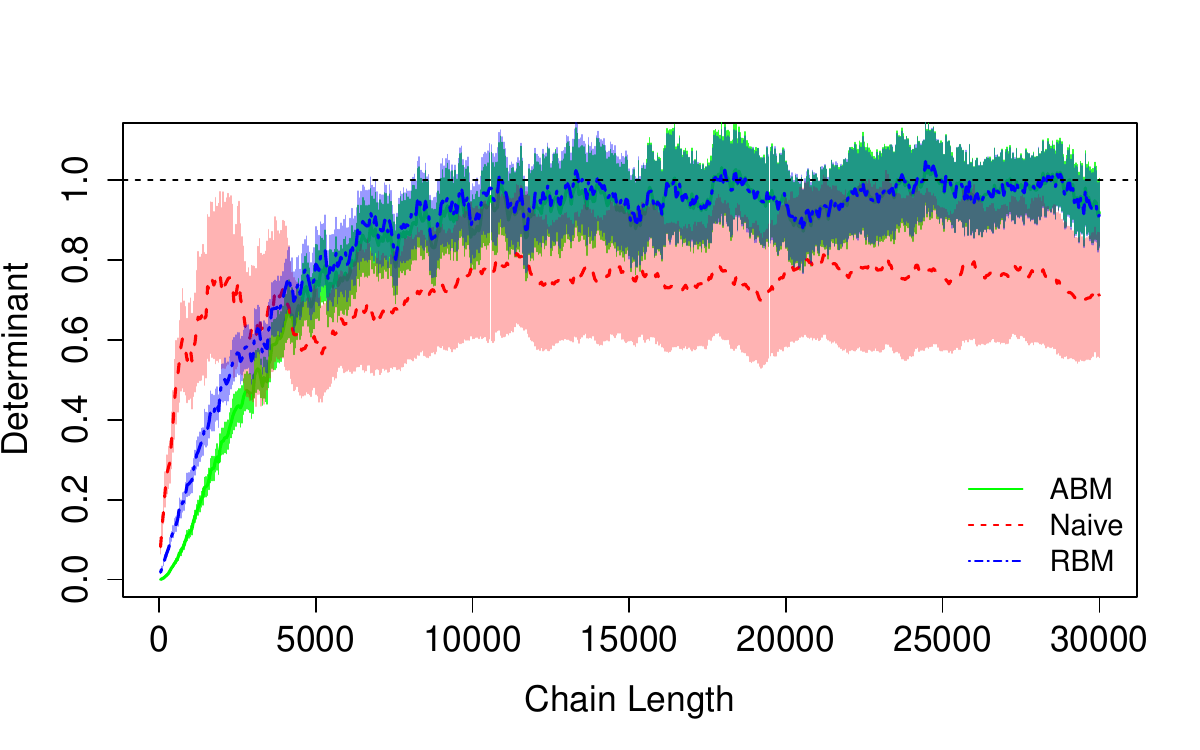}
         \label{subfig:MVG_running_rho_0.999}
}
 \\ \vspace{-1cm}
    \subfloat[][$m = 10, \rho = 0.5$]{
         \centering
         \includegraphics[width = 2.9in]{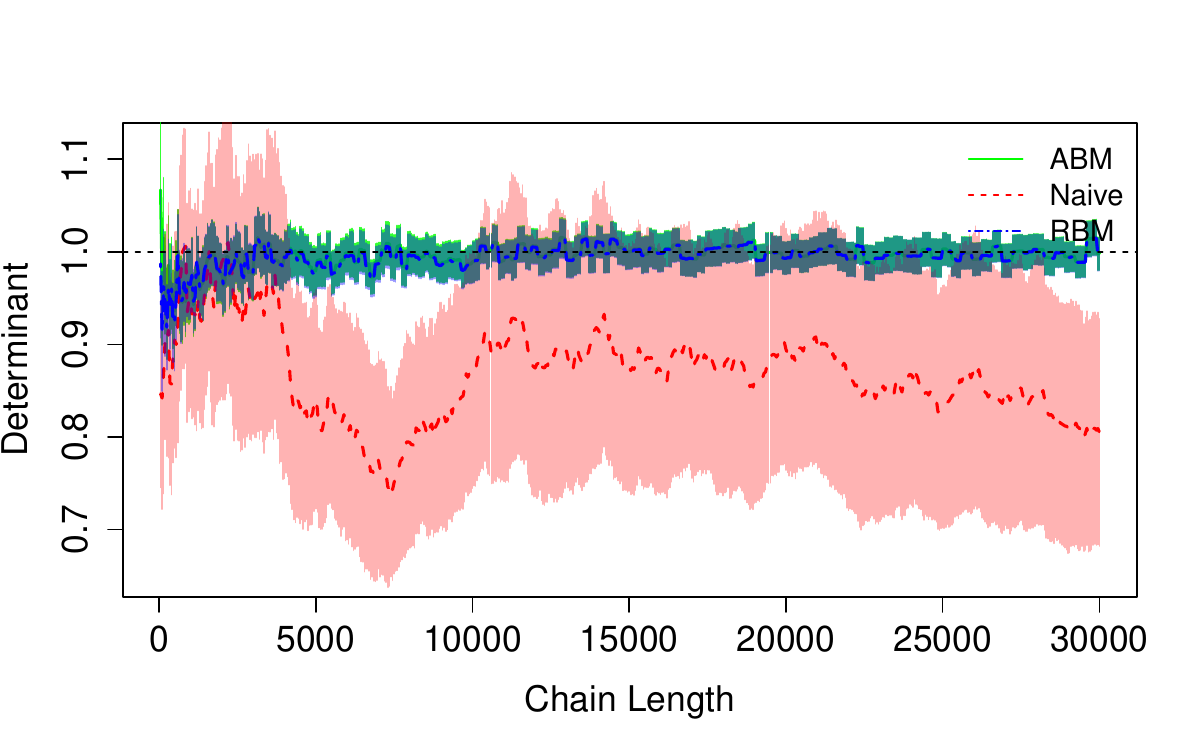}
         \label{subfig:MVG_running_m_10_rho_0.5}
}
    \subfloat[][$m = 10, \rho = 0.999$]{
         \centering
         \includegraphics[width = 2.9in]{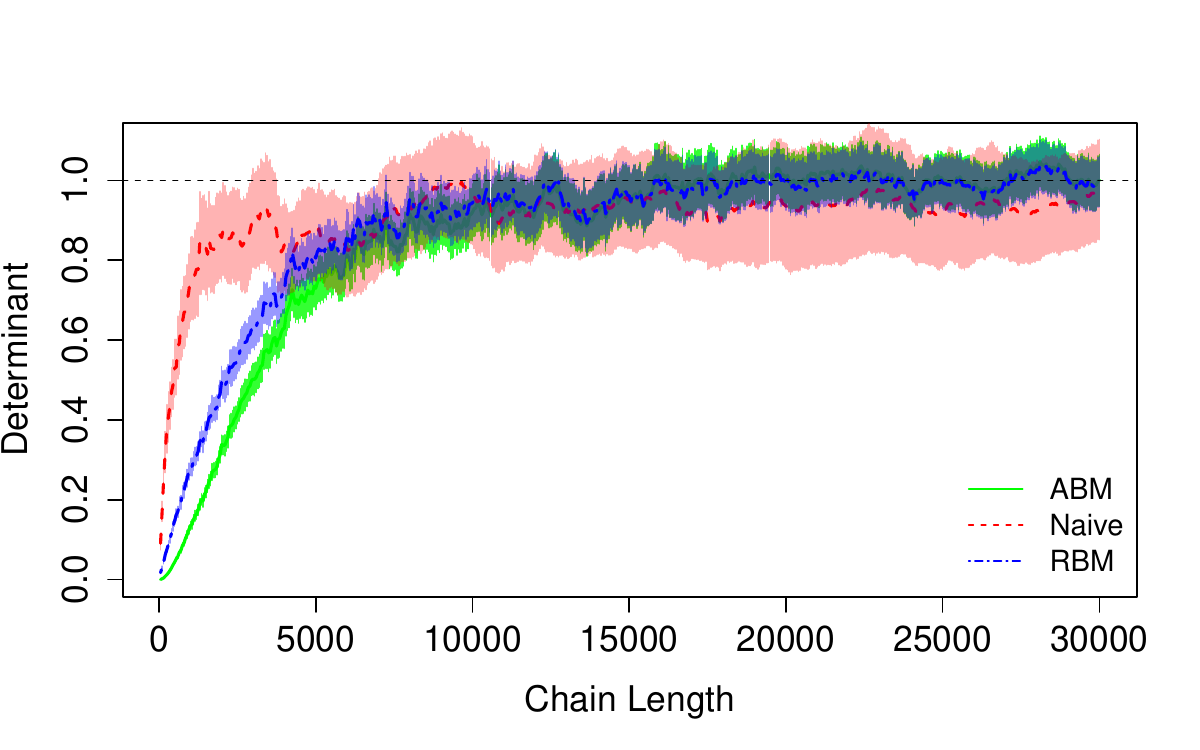}
         \label{subfig:MVG_running_m_10_rho_0.999}
}
    \caption{Bivariate normal: Determinant running plots (with standard errors) over 100 replications.}
    \label{fig:MVG_FN_running}
\end{figure}
To quantify the effect of the quality of estimation of $\Sigma$, we estimate the coverage probabilities of the resulting $95\%$ confidence regions over 1000 replications for both settings. Note that coverage probability is the proportion of confidence intervals that contain the parameter of interest. In addition to the coverage probabilities obtained using ABM, naive, and RBM, we also estimate the coverage probability using the true $\Sigma$; this, in some sense, represents the oracle. The results are presented in Table~\ref{table:MVG_coverage_m_5_10}. For Setting 1 (for both $m = 5,10$), ABM and RBM yield similar coverage, as was expected. Here, the high variability of the naive estimator impacts the coverage significantly which does not improve as $n$ increases, demonstrating the effect of Theorem~\ref{thm:naive_convergence}. For Setting 2, there is a clear separation of coverage probabilities between RBM and ABM, especially for smaller sample sizes. As the sample size increases, ABM and RBM yield more similar coverage as a consequence of less separation between the means from each chain. 
\begin{table}[htbp]
\renewcommand{\arraystretch}{1}
\centering
    \subfloat[][$m = 5, \rho = 0.5$]{
    \begin{tabular}{|c|c|c|c|c|}
    \hline
        $n$ & ABM & Naive & RBM & True \\
        \hline
        $5e2$ & 0.930 & 0.752 & 0.929 & 0.966 \\
        \rowcolor{Gray}
        $1e3$ & 0.944 & 0.767 & 0.947 & 0.958 \\
        $5e3$ & 0.952 & 0.756 & 0.952 & 0.957 \\
        \rowcolor{Gray}
        $3e4$ & 0.954 & 0.736 & 0.954 & 0.958 \\ \hline
    \end{tabular}         \label{subtable:MVG_coverage_rho_0.5_m_5}
    }
    \qquad
    \subfloat[][$m = 5, \rho = 0.999$]{
    \begin{tabular}{|c|c|c|c|c|} 
    \hline
        $n$ & ABM & Naive & RBM & True \\
        \hline
        $5e2$ & 0.367 & 0.755 & 0.602 & 0.945 \\
        \rowcolor{Gray}
        $1e3$ & 0.536 & 0.745 & 0.677 & 0.949 \\
        $5e3$ & 0.838 & 0.753 & 0.864 & 0.950 \\
        \rowcolor{Gray}
        $3e4$ & 0.926 & 0.755 & 0.922 & 0.954 \\ \hline
    \end{tabular}    \label{subtable:MVG_coverage_rho_0.999_m_5}
    } \\ 

    \subfloat[][$m = 10, \rho = 0.5$]{
    \begin{tabular}{|c|c|c|c|c|}
    \hline
        $n$ & ABM & Naive & RBM & True \\
        \hline
        $5e2$ & 0.944 & 0.862 & 0.941 & 0.942 \\
        \rowcolor{Gray}
        $1e3$ & 0.947 & 0.879 & 0.945 & 0.948 \\
        $5e3$ & 0.938 & 0.860 & 0.939 & 0.938 \\
        \rowcolor{Gray}
        $3e4$ & 0.945 & 0.869 & 0.946 & 0.947 \\ \hline
    \end{tabular}         \label{subtable:MVG_coverage_rho_0.5_m_10}
    }
    \qquad
    \subfloat[][$m = 10, \rho = 0.999$]{
    \begin{tabular}{|c|c|c|c|c|} 
    \hline
        $n$ & ABM & Naive & RBM & True \\
        \hline
        $5e2$ & 0.418 & 0.887 & 0.678 & 0.940 \\
        \rowcolor{Gray}
        $1e3$ & 0.538 & 0.866 & 0.735 & 0.951 \\
        $5e3$ & 0.889 & 0.869 & 0.911 & 0.947 \\
        \rowcolor{Gray}
        $3e4$ & 0.932 & 0.856 & 0.931 & 0.942 \\ \hline
    \end{tabular}    \label{subtable:MVG_coverage_rho_0.999_m_10}
    }
    \caption{Bivariate normal: Coverage probabilities from 1000 replications at 95\% nominal level.}
    \label{table:MVG_coverage_m_5_10}
\end{table}

In the rest of the examples, we focus only on comparing the RBM to the ABM estimator, since it is evident that the naive estimator is unreliable. 

\subsection{Rosenbrock distribution}

The Rosenbrock distribution is commonly used to serve as a benchmark for testing MCMC algorithms. As shown in Figure~\ref{fig:rosenbrock_contour}, the density takes positive values along a narrow parabola making it difficult for MCMC algorithms to take large steps. Consider the 2-dimensional Rosenbrock density discussed in \cite{goodman2010} and \cite{pagani2019ndimensional}:
\[
\pi (x_{1}, x_{2}) \propto \exp \left\{-\dfrac{1}{20}(x_{1} - 1)^{2} - 5 \left(x_{2} - x_{1}^{2} \right)^{2} \right\} \,.
\]
Although there are specialized MCMC algorithms available that traverse the contours of the Rosenbrock density more efficiently, we implement a random walk Metropolis-Hastings algorithm with a Gaussian proposal.  Trace plots for one component of two runs are presented in Figure~\ref{fig:rosenbrock_contour}.
\begin{figure}[htbp]
    \centering
    \subfloat[][$2d$ Rosenbrock density]{    
    \includegraphics[width= 2.8in] {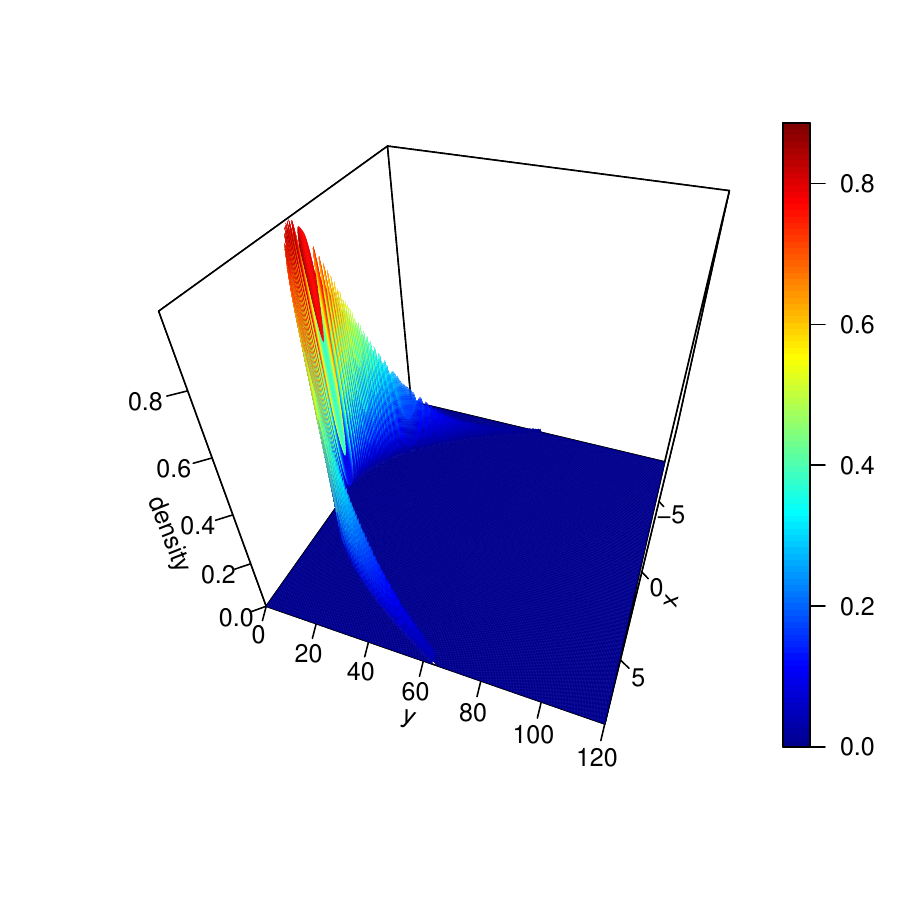}
    }
  \subfloat[][Trace plots for two Markov chains]{
    \includegraphics[width= 2.8in] {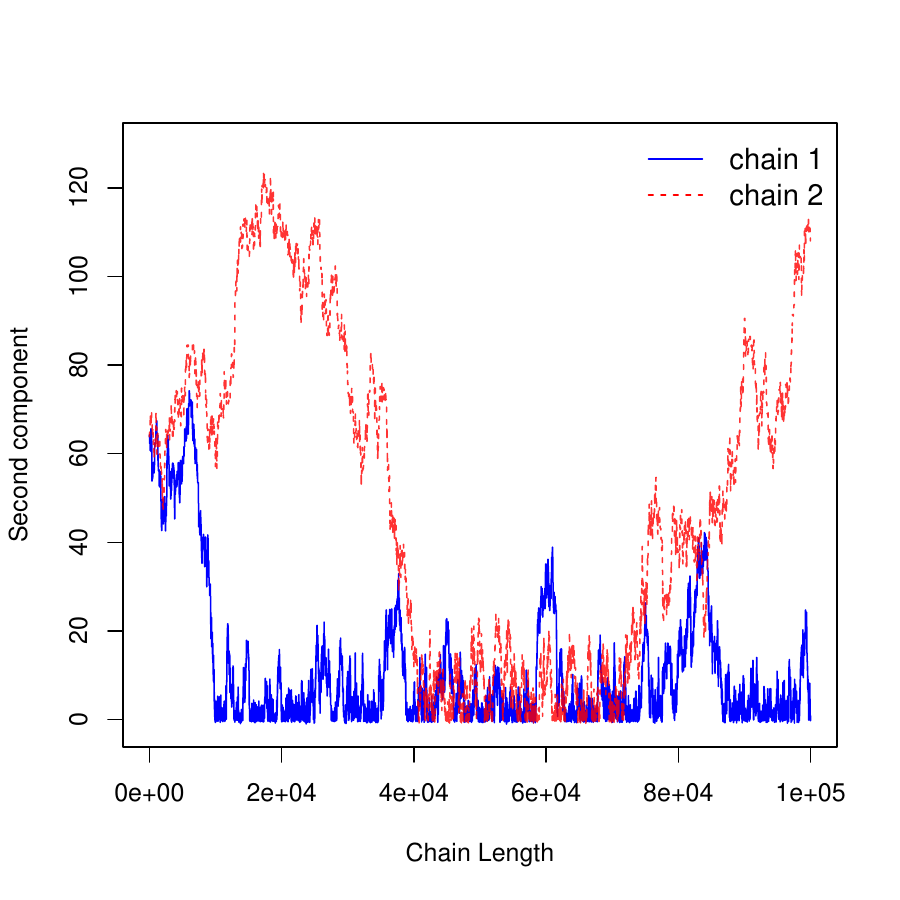}
      }
    \caption{$2d$ Rosenbrock density and trace plots.}
    \label{fig:rosenbrock_contour}
\end{figure}

We set $m = 5,10$ with starting points sampled from the target. Since $\Sigma$ is unknown, we compare ABM and RBM by analyzing running plots of the estimated ESS. Simulation terminates when an estimated ESS, $\widehat{\text{ESS}} > M$, for a pre-specified  $M$. To avoid early termination, it is critical that the estimate of $\Sigma$ is not underestimated.

\begin{figure}[htbp]
    \centering
    \subfloat[][$m = 5$]{
         \centering
         \includegraphics[width=2.4in]{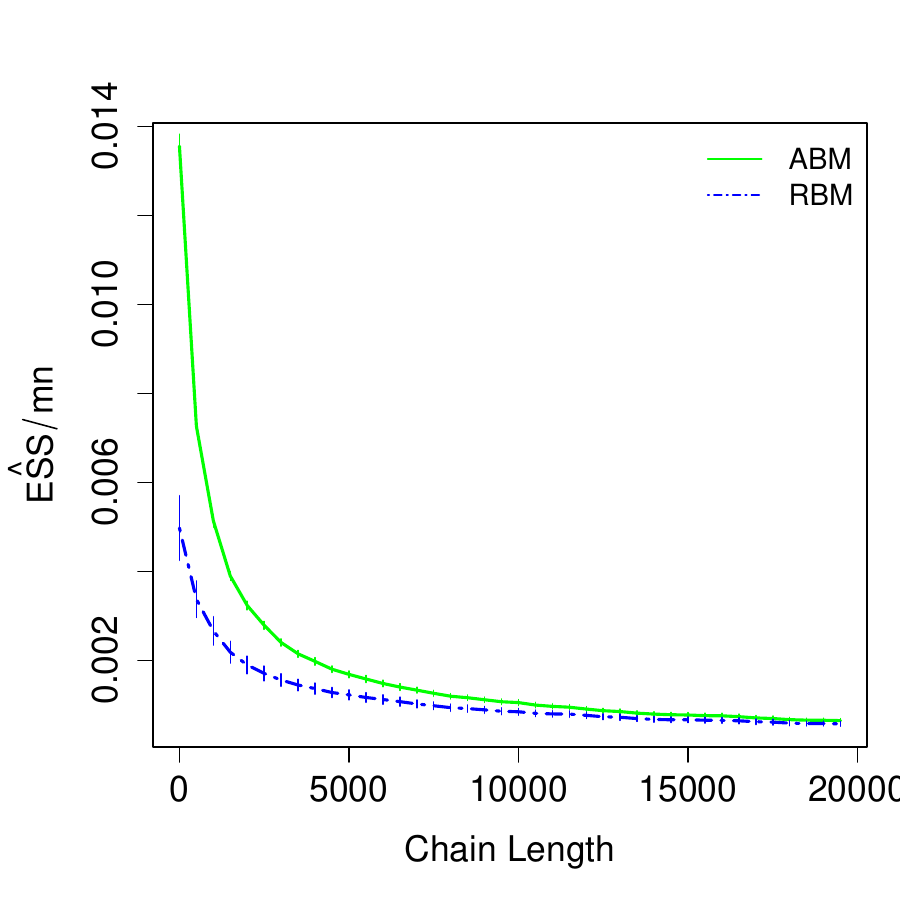}
         \label{subfig:rosenbrock_ESS_running_m_5}
}
    \subfloat[][$m = 10$]{
         \centering
         \includegraphics[width=2.4in]{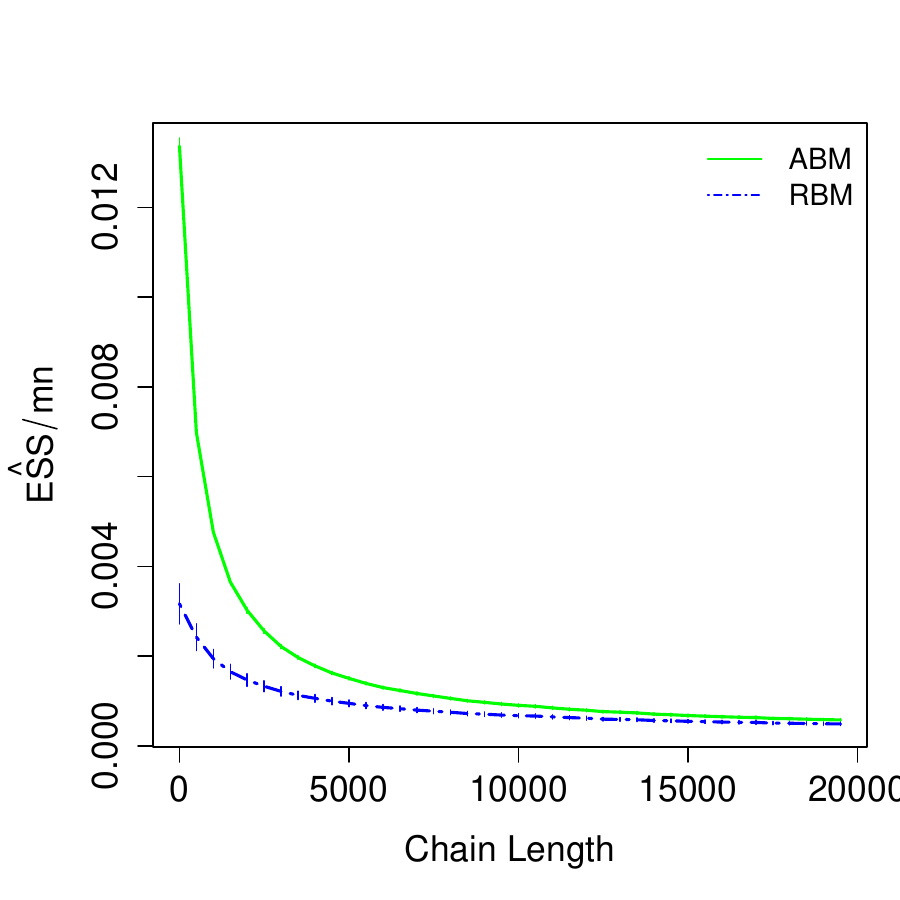}
         \label{subfig:rosenbrock_ESS_running_m_10}
}
    \caption{Rosenbrock: $\widehat{\text{ESS}}/mn$ running plots (with standard errors).}
    \label{fig:rosenbrock_ESS_running}
\end{figure}
Figure~\ref{fig:rosenbrock_ESS_running} presents the estimated $\widehat{\text{ESS}}/mn$ over 100 replications using the average sample covariance matrix estimate of $\Var_F(h(X_{11}))$ and ABM and RBM estimates of $\Sigma$. We observe that for both $m = 5,10$, ABM grossly overestimates $\widehat{\text{ESS}}/mn$ and converges from above, whereas RBM converges from below. As a consequence, using RBM will help safeguard against early termination, whereas ABM is significantly more likely to produce inadequate estimates at termination.

Since the true mean for this target density is known, we compare coverage probabilities of the resulting confidence regions using ABM and RBM over 1000 replications in Table~\ref{table:rosenbrock_coverage_m_5}. The results for $m = 5$ and $m = 10$ are similar; the ABM estimator yields abysmally low coverage, especially at the beginning of the process. The RBM estimator, on the other hand, yields a high coverage probability despite the slow mixing nature of the Markov chain. Although the coverage is below the desired $95\%$ mark, we observe that RBM gives better coverage than ABM till chains of length = $10^5$. Given the slow mixing nature of the chain, we can expect the coverage to increase for both estimators with an increase in running length.

\begin{table}[h]
\renewcommand{\arraystretch}{1}
\centering
    \subfloat[$m = 5$]{
    \begin{tabular}{|c|c|c|}
    \hline
        $n$ & ABM & RBM \\
        \hline
        $5e3$ & 0.251 & 0.364 \\
        \rowcolor{Gray}
        $10^4$ & 0.378 & 0.449\\
        $5e4$ & 0.604 & 0.621 \\
        \rowcolor{Gray}
        $10^5$ & 0.706 & 0.721 \\ \hline
    \end{tabular}
    }
    \qquad
    \subfloat[$m = 10$]{
    \begin{tabular}{|c|c|c|}
    \hline
        $n$ & ABM & RBM \\
        \hline
        $5e3$ & 0.265 & 0.424 \\
        \rowcolor{Gray}
        $10^4$ & 0.391 & 0.498\\
        $5e4$ & 0.714 & 0.742 \\
        \rowcolor{Gray}
        $10^5$ & 0.763 & 0.776 \\ \hline
    \end{tabular}
    }
    \caption{Rosenbrock: Coverage probabilities over 1000 replications at 95\% nominal level.}
    \label{table:rosenbrock_coverage_m_5}
\end{table}

\subsection{Bayesian neural networks}

Generating MCMC samples from the parameter posteriors of neural networks is a  notoriously challenging problem. This example considers a multi-layer perceptron (MLP) with three neurons in the output layer performing multiclass classification on a real-world dataset involving penguins. This dataset consists of body measurements for three penguin species: Adelie, Chinstrap, and Gentoo penguins. Four body measurements per penguin, specifically body mass, flipper length, bill length, and bill depth, along with the sex and location (island), are used to predict the penguin species. More information about the dataset can be found in \cite{penguins_data}. 

Following the experiments from \cite{papamarkou2021challenges}, we use the MLP model for the penguin example to get 29 MLP parameters. We employ the \texttt{bnn\_mcmc\_examples} package to generate the 29 dimensional Hamiltonian Monte Carlo (HMC) chains on the MLP parameters with the same specifications as used in the paper. We run 5 parallel chains of size 1e5 each with starting values sampled from the Gaussian prior and repeat for 10 replications. The trace plots of dimensions 15 and 23 in figure \ref{fig:penguins_trace} show the multi-modal nature of the target, with the Markov chain jumping between the modes. 

\begin{figure}[htbp]
    \centering
\subfloat[][$\beta_{15}$]{
         \centering
         \includegraphics[width=.44\textwidth]{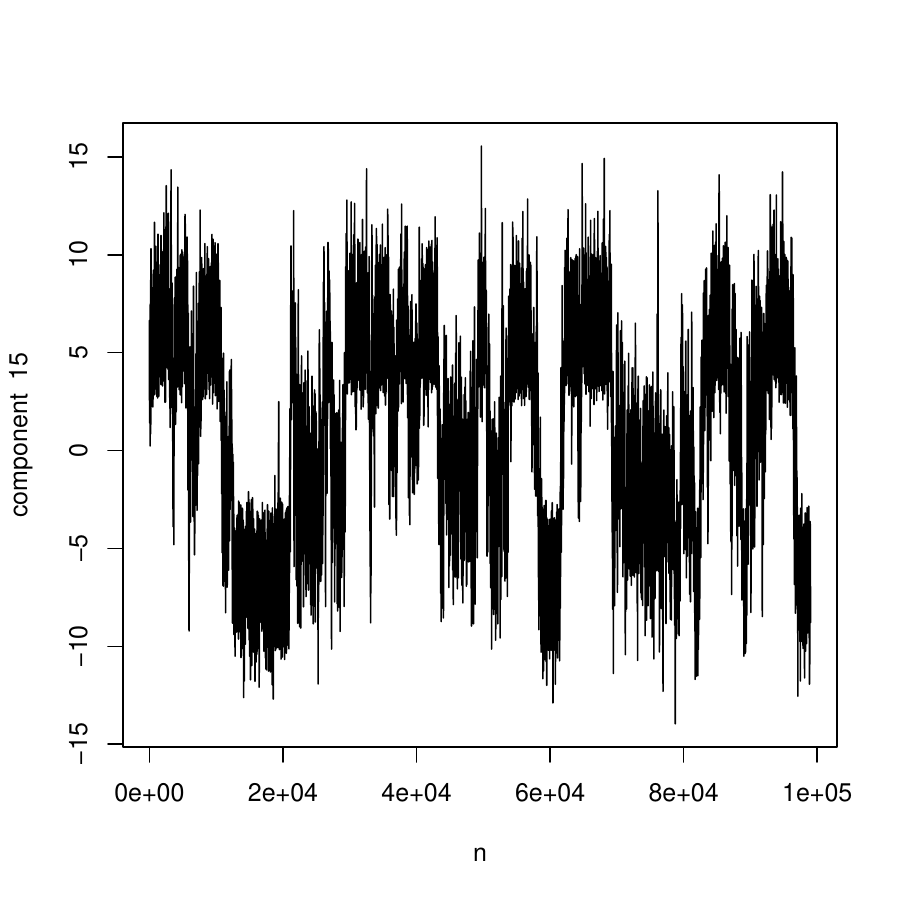}
        }
     \qquad
\subfloat[][$\beta_{23}$]{
         \centering
         \includegraphics[width=.44\textwidth]{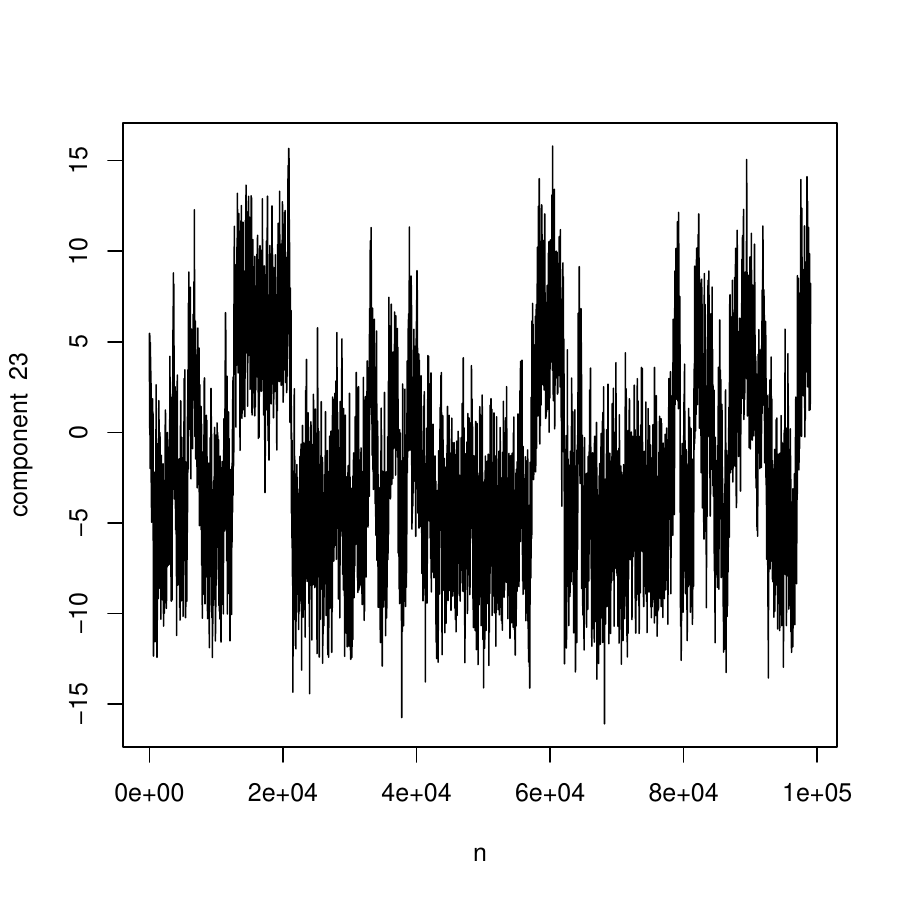}
        }
     \caption{Bayesian NN: Trace plots for multi-modal components.}
    \label{fig:penguins_trace}
\end{figure}

In this example, both the true posterior mean and the resulting asymptotic covariance matrix $\Sigma$ are unknown. Thus coverage probabilities are not estimable. We, therefore, focus on the quality of estimation of ESS. Figure~\ref{fig:penguin_ESS_running} contains the running estimate of $\widehat{\text{ESS}}/mn$ using both ABM and RBM estimators. 

\begin{figure}[htbp]
    \centering
    \includegraphics[width= 2.5in]{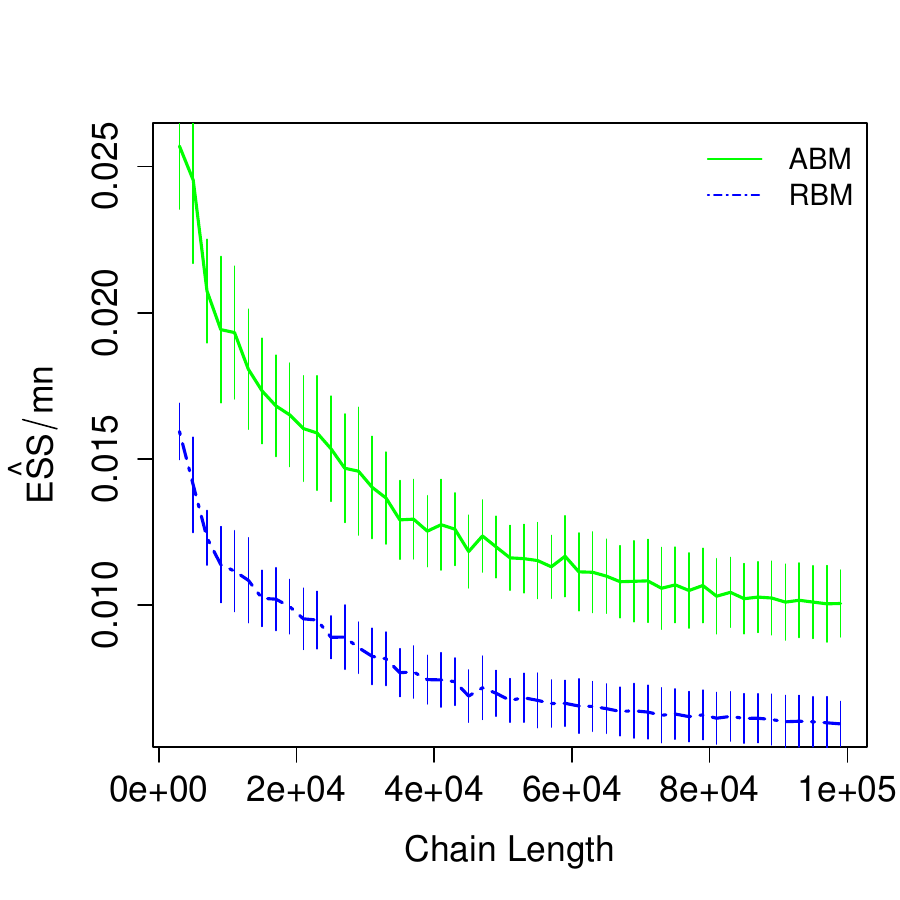}
    \caption{Bayesian NN: $\widehat{\text{ESS}}/mn$ running plots averaged over 10 replications.}
    \label{fig:penguin_ESS_running}
\end{figure}

Given the multi-modal nature of the target, the Markov chains take a while to explore the state space adequately. The advantage of using RBM over ABM in such scenarios is clearly evident. RBM produces a systematically lower estimate of ESS, whereas the ABM estimator may give a false sense of security in early small runs of the sampler. This goes to show that the finite sample gains of using the RBM estimator are undeniable, and MCMC diagnostics are much more reliable with RBM.

\section{Discussion}

Work in MCMC output analysis has so far not caught up to the ease with which users can implement parallel MCMC runs. This work is an attempt at consistent estimation of $\Sigma$ for parallel MCMC.   For fast mixing Markov chains, the RBM estimator is as good as the current state-of-the-art, and for slow mixing chains, the RBM estimator provides significant improvement and safeguards users from premature termination of the process.  \cite{agarwal2022globally} proposed a parallel chain version of the spectral variance estimator, but as is widely understood, batch-means estimator remains to be significantly more computationally efficient.

Choosing a batch size is a challenging problem for single-chain batch means, and the optimal batch size of \cite{liu2018batch} is a step in the right direction. Since each realization of a Markov chain can yield different optimal batch sizes, an important future extension would be to construct RBM estimators from pooling batches of different batch sizes, each tuned to its own Markov chain. Additionally, here we consider only non-overlapping batch means estimators due to their computational feasibility. Similar construction for multiple-chain output analysis should be possible for the overlapping batch means estimators found in \cite{fleg:jone:2010,liu2018batch}.


\section{Acknolwedgements}

Dootika Vats is supported by SERB (SPG/2021/001322).

\appendix

\section{Preliminaries}
We present a few preliminary results to assist us in our proofs.
\begin{lemma}
\label{lemma: brownian_bound}
(Csorgo and Revesz (1981)). Suppose Assumption\ref{ass:bn} holds, then for all $\epsilon > 0$ and for almost all sample paths, there exists $n_{0}\left(\epsilon\right)$ such that $\forall n\geq n_{0}$ and $\forall i = 1, ..., p$
\[
\sup_{0\leq t \leq n-b_n}\sup_{0 \leq s \leq b_n} \left| B^{\left(i\right)}\left(t+s\right) - B^{\left(i\right)}\left(t\right) \right| < \left(1+ \epsilon\right)\left(2b_n\left(\log\dfrac{n}{b_n} + \log\; \log\; n\right)\right)^{1/2} ,
\]
\[
\sup_{0 \leq s \leq b_n} \left|B^{\left(i\right)}\left(n\right) - B^{\left(i\right)}\left(n - s\right)\right| < \left(1+ \epsilon\right)\left(2b_n\left(\log\dfrac{n}{b_n} + \log\;\log\;n\right)\right)^{1/2} , \;and
\]
\[
\left|B^{\left(i\right)}\left(n\right)\right| < \left(1+\epsilon\right)\sqrt{2n\;\log \log n} \; . 
\]
\end{lemma}


The following straightforward decomposition will be used often and thus is presented as a lemma.
\begin{lemma}
  \label{lem:Y_diff_simpli}
  For $\hat{\mu}_k$ and $\hat{\mu}$ defined as before,
\[
  \sum\limits_{k=1}^{m}\left(\hat{\mu}_k-\hat{\mu}\right)\left(\hat{\mu}_k-\hat{\mu}\right)^{T} = \ds \sum_{k=1}^{m} (\hat{\mu}_k - \mu)(\hat{\mu}_k - \mu)^T - m (\hat{\mu} - \mu)(\hat{\mu} - \mu)^T\,.
\]
\end{lemma}

\subsection{Strong consistency of RBM}
\label{subsec: strong_consistency_RBM}

\begin{proof}[Proof of Theorem~\ref{thm:RBM_consistency}]
We will show that the RBM estimator can be decomposed into the ABM estimator, plus some small order terms. Consider 
\begin{align*}
\hat{\Sigma}_{R,b_n} &= \dfrac{b_n}{a_nm-1}\sum_{k=1}^{m}\sum_{s=1}^{a_n}\left(\bar{Y}_{k,s} - \hat{\mu}\right)\left(\bar{Y}_{k,s} - \hat{\mu}\right)^{T}\\
&= \dfrac{b_n}{a_nm-1}\sum_{k=1}^{m}\sum_{s=1}^{a_n}\left(\bar{Y}_{k,s} - \hat{\mu}_k + \hat{\mu}_k - \hat{\mu}\right)\left(\bar{Y}_{k,s} - \hat{\mu}_k + \hat{\mu}_k - \hat{\mu}\right)^{T}\\
&= \dfrac{b_n}{a_nm-1}\sum_{k=1}^{m}\sum_{s=1}^{a_n}\left[\left(\bar{Y}_{k,s}-\hat{\mu}_k\right)\left(\bar{Y}_{k,s}-\hat{\mu}_k\right)^{T} + \left(\bar{Y}_{k,s}-\hat{\mu}_k \right)\left(\hat{\mu}_k - \hat{\mu}\right)^{T}  \right.\\
& \quad \quad  \left. + \left(\hat{\mu}_k-\hat{\mu}\right)\left(\bar{Y}_{k,s}-\hat{\mu}_k\right)^{T} + \left(\hat{\mu}_k-\hat{\mu}\right)\left(\hat{\mu}_k-\hat{\mu}\right)^{T}\right]\\
&= \dfrac{b_n}{a_nm-1}\dfrac{a_n-1}{b_n}\sum_{k=1}^{m}\hat{\Sigma}_{k,b_n} + \dfrac{b_n}{a_nm-1}\sum_{k=1}^{m}\left[\sum_{s=1}^{a_n}\left(\bar{Y}_{k,s}-\hat{\mu}_k\right)\right]\left(\hat{\mu}_k-\hat{\mu}\right)^{T} \\
& \quad \quad + \dfrac{b_n}{a_nm-1}\sum_{k=1}^{m}\left(\hat{\mu}_k-\hat{\mu}\right)\sum_{s=1}^{a_n}\left(\bar{Y}_{k,s}-\hat{\mu}_k\right)^{T} + \dfrac{a_nb_n}{a_nm-1}\sum_{k=1}^{m}\left(\hat{\mu}_k-\hat{\mu}\right)\left(\hat{\mu}_k-\hat{\mu}\right)^{T}\\
&= \dfrac{m(a_n-1)}{a_nm-1} \hat{\Sigma}_{A,b_n} + \dfrac{a_nb_n}{a_nm-1}\sum_{k=1}^{m}\left(\hat{\mu}_k-\hat{\mu}\right)\left(\hat{\mu}_k-\hat{\mu}\right)^{T} \,. \numberthis \label{eq:repli_breakdown}
\end{align*}

By Lemma~\ref{lem:Y_diff_simpli}
\begin{align*}
\dfrac{a_nb_n}{a_nm - 1}\sum_{k=1}^{m}\left(\hat{\mu}_k-\hat{\mu}\right)\left(\hat{\mu}_k-\hat{\mu}\right)^{T} 
&= \dfrac{a_nb_n}{a_nm - 1}\sum_{k=1}^{m}\left(\hat{\mu}_k-\mu\right)\left(\hat{\mu}_k-\mu\right)^{T} - \dfrac{a_n b_n m}{a_nm - 1}\left(\hat{\mu}-\mu\right)\left(\hat{\mu}-\mu\right)^{T}
\end{align*}

Let 
\[
A = \dfrac{a_n}{a_nm - 1}\sum_{k=1}^{m}b_n\left(\hat{\mu}_k-\mu\right)\left(\hat{\mu}_k-\mu\right)^{T} \quad \text{ and }  \quad B = \dfrac{a_n b_n m}{a_nm - 1}\left(\hat{\mu}-\mu\right)\left(\hat{\mu}-\mu\right)^{T}\,.
\]
By \cite{stra:1964}, $\psi(n) = O(\sqrt{n \log \log n})$ so that the rate is such that the law of iterated logarithm holds. Let $A^{ij}$, be the $(i,j)$th element of the matrix $A$. Using the SIP and Lemma~\ref{lemma: brownian_bound},
\begin{align*}
    \left |A^{ij} \right | &= \dfrac{a_n}{a_nm - 1}\sum_{k=1}^{m} \left| b_n\left(\bar{Y}^{\left(i\right)}_{k} - \mu^{(i)}\right)\left(\bar{Y}^{\left(j\right)}_{k} - \mu^{(j)}\right) \right|\\
    & \leq \dfrac{a_n}{a_nm - 1}\sum_{k=1}^{m}\left[ \left| \dfrac{b_n}{n^{2}}\left(\sum_{r=1}^{n}Y^{\left(i\right)}_{r,k} - n\mu^{(i)}\right)\left(\sum_{s=1}^{n}Y^{\left(j\right)}_{s,k} - n\mu^{(j)}\right) \right|  \right]\\
    &= \dfrac{a_n}{a_nm - 1}\sum_{k=1}^{m}\Bigg[\dfrac{b_n}{n^{2}}  \Bigg|\left(\sum_{r=1}^{n}Y^{\left(i\right)}_{r,k} - n\mu^{(i)} - \Sigma_{ii}B^{\left(i\right)} + \Sigma_{ii}B^{\left(i\right)}\right) \\ 
    & \qquad \qquad \qquad \qquad \times\left(\sum_{s=1}^{n}Y^{\left(j\right)}_{s,k} - n\mu^{(j)} - \Sigma_{jj}B^{\left(j\right)} + \Sigma_{jj}B^{\left(j\right)}\right) \Bigg| \Bigg]\\
    & \leq \dfrac{a_n}{a_nm - 1}\sum_{k=1}^{m}\left[\dfrac{b_n}{n^{2}}\left[ \left| \left(\sum_{r=1}^{n}Y^{\left(i\right)}_{r,k} - n\mu^{(i)} - \Sigma_{ii}B^{\left(i\right)}\right)\left(\sum_{s=1}^{n}Y^{\left(j\right)}_{s,k} - n\mu^{(j)} - \Sigma_{jj}B^{\left(j\right)}\right) \right| \right. \right.\\
    & \left. \left. \quad + \left|\left(\Sigma_{ii}B^{\left(i\right)}\right) \right|    \left|\left(\sum_{s=1}^{n}Y^{\left(j\right)}_{s,k} - n\mu^{\left(j\right)} - \Sigma_{jj}B^{\left(j\right)}\right)\right|   +   \left|\left(\Sigma_{jj}B^{\left(j\right)}\right) \right|   \left|\left(\sum_{r=1}^{n}Y^{\left(i\right)}_{r,k} - n\mu^{(i)} - \Sigma_{ii}B^{\left(i\right)}\right) \right| \right. \right.\\
    &  \quad + \left|\left(\Sigma_{ii}B^{\left(i\right)}\right) \right|  \left| \left(\Sigma_{jj}B^{\left(j\right)}\right) \right| \Bigg]\Bigg]\\
    &< \dfrac{a_n}{a_nm - 1}\sum_{k=1}^{m}\left[\dfrac{b_n}{n^{2}}\left[\left(D\psi\left(n\right)\right)^{2} + \Sigma_{ii}\left(1 + \epsilon\right)\sqrt{2n\log\log n}\left(D\psi\left(n\right)\right)  \right. \right.\\
    & \quad \left. + \Sigma_{jj}\left(1 + \epsilon\right)\sqrt{2n\log\log n}\left(D\psi\left(n\right)\right) + \Sigma_{ii}\Sigma_{jj}\left(1 + \epsilon\right)^{2}\left(2n\log\log n\right) \right] \Bigg]\\
    &= \dfrac{a_nm}{a_nm - 1} \left[\dfrac{b_n}{n}\left(\dfrac{D \psi\left(n\right)}{\sqrt{n}}\right)^{2} + \left(\Sigma_{ii} + \Sigma_{jj}\right)\left(1 + \epsilon\right)\dfrac{b_n \sqrt{2\log\log n}}{n}\left(\dfrac{D\psi\left(n\right)}{\sqrt{n}}\right) \right.\\
    & \quad + \left. 2\Sigma_{ii}\Sigma_{jj}\left(1 + \epsilon\right)^{2}\left(\dfrac{b_n\log\log n}{n}\right) \right]\\
     &= \dfrac{a_nm}{a_nm - 1} \dfrac{b_n\log \log n}{n} \\
     & \quad \times\left[\left(\dfrac{D \psi\left(n\right)}{\sqrt{n \log \log n}}\right)^{2} + \sqrt{2} \left(\Sigma_{ii} + \Sigma_{jj}\right)\left(1 + \epsilon\right)\left(\dfrac{D\psi\left(n\right)}{\sqrt{n \log \log n}}\right) +  2\Sigma_{ii}\Sigma_{jj}\left(1 + \epsilon\right)^{2}\right]\\
    & \xrightarrow{a.s} 0 \; \; n \to \infty \, .
\end{align*}
Following the same steps for $B^{ij}$ with $n$ replaced by $mn$,  $\left|B^{ij} \right| \xrightarrow{a.s} 0 \; \; n \to \infty$.
Therefore $\|A\| \xrightarrow{a.s} \mathbf{0}$ and  $\|B\| \xrightarrow{a.s} \mathbf{0}$ as  $n \to \infty$. Using this,
\begin{equation}
\label{eq:Y_diff_term_0}
  \left\|\dfrac{a_nb_n}{a_nm - 1}\sum_{k=1}^{m}\left(\hat{\mu}_k-\hat{\mu}\right)\left(\hat{\mu}_k-\hat{\mu}\right)^{T}\right\| = \left\|A - B \right\| \leq \|A\| + \|B\| \xrightarrow{a.s} \mathbf{0} \; \; n \to \infty \,.
\end{equation}

Using \eqref{eq:Y_diff_term_0} and consistency of $\hat{\Sigma}_{k,b_n}$ for all $k$ in \eqref{eq:repli_breakdown}
\begin{align*}
  \left\|\hat{\Sigma}_{R,b_n} - \Sigma   \right\| & =  \left\|\dfrac{m(a_n-1)}{a_nm-1} \hat{\Sigma}_{A,b_n} - \Sigma +   \dfrac{a_nb_n}{a_nm-1}\sum_{k=1}^{m}\left(\hat{\mu}_k-\hat{\mu}\right)\left(\hat{\mu}_k-\hat{\mu}\right)^{T}\right\| \\
  & \leq \left\| \dfrac{m(a_n-1)}{a_nm-1} \hat{\Sigma}_{A,b_n} - \Sigma \right\| +  \left\|  \dfrac{a_nb_n}{a_nm-1}\sum_{k=1}^{m}\left(\hat{\mu}_k-\hat{\mu}\right)\left(\hat{\mu}_k-\hat{\mu}\right)^{T}  \right\|\\
  & \overset{a.s.}{\to} 0 \quad \text{ as } n \to \infty\,.
\end{align*}
\end{proof}



\subsection{Bias of RBM}
\label{subsec:bias_RBM}

\begin{proof}[Proof of Theorem~\ref{thm:rbm_bias}]
First, using Lemmas 1-3 of \citep{song:schm:1995},
\begin{align*}
\mathbb{E}\left[\sum\limits_{k=1}^{m}\left(\hat{\mu}_k-\hat{\mu}\right)\left(\hat{\mu}_k-\hat{\mu}\right)^{T}\right] &= \mathbb{E}\left[\sum_{k=1}^{m}\hat{\mu}_k\hat{\mu}_k^{T} -\sum_{k=1}^{m}\hat{\mu}_k\hat{\mu}^{T} - \hat{\mu}\sum_{k=1}^{m}\hat{\mu}_k^{T} + m\hat{\mu}\hat{\mu}^{T}\right]\\
&= \mathbb{E}\left[\sum_{k=1}^{m}\hat{\mu}_k\hat{\mu}_k^{T} - m\hat{\mu}\hat{\mu}^{T}\right]\\
&= m\left(\mathbb{E}\left[\bar{Y}_{1}\bar{Y}_{1}^{T}\right] - \mathbb{E}\left[\hat{\mu}\hat{\mu}^{T}\right]\right)\\
&= m\left(\Var\left(\bar{Y}_{1}\right) - \Var\left(\hat{\mu}\right)\right)\\
&= m\left[\dfrac{\Sigma}{n} + \dfrac{\Gamma}{n^{2}} + o\left(\dfrac{1}{n^{2}}\right) - \dfrac{\Sigma}{mn} - \dfrac{\Gamma}{mn^{2}} + o\left(\dfrac{1}{mn^{2}}\right)\right]\\
&= m\left[\dfrac{\Sigma}{n}\left(1-\dfrac{1}{m}\right) + \dfrac{\Gamma}{n^{2}}\left(1-\dfrac{1}{m}\right) + o\left(\dfrac{1}{n^{2}}\right)\right]\,. \numberthis \label{eq:expect_y_diff}
\end{align*}
By \eqref{eq:repli_breakdown}, \eqref{eq:expect_y_diff}, and the bias of the  BM estimator from  \citet[Theorem 1]{vats:flegal:2018},
\begin{align*}
& \mathbb{E}[\hat{\Sigma}_{R,b_n}] \\
&= \mathbb{E}\left[\dfrac{a_n-1}{a_nm-1}\sum_{k=1}^{m}\hat{\Sigma}_{k,b_n} + \dfrac{a_nb_n}{a_nm-1}\sum_{k=1}^{m}\left(\hat{\mu}_k-\hat{\mu}\right)\left(\hat{\mu}_k-\hat{\mu}\right)^{T}\right]\\
&= \dfrac{m\left(a_n-1\right)}{a_nm-1}\mathbb{E}\left[\hat{\Sigma}_{k,1}\right] + \dfrac{a_n b_n m}{a_nm-1}\left(\dfrac{\Sigma}{n}\left(1-\dfrac{1}{m}\right) + \dfrac{\Gamma}{n^{2}}\left(1-\dfrac{1}{m}\right) + o\left(\dfrac{1}{n^{2}}\right)\right)\\
&= \dfrac{m\left(a_n-1\right)}{a_nm-1}\left[\Sigma + \dfrac{\Gamma}{b_n} + o\left(\dfrac{1}{b_n}\right)\right] + \dfrac{m\Sigma}{a_nm-1}\left(1-\dfrac{1}{m}\right) + \dfrac{m\Gamma}{n\left(a_nm-1\right)}\left(1-\dfrac{1}{m}\right)\\
& \qquad + \dfrac{m}{a_nm-1}o\left(\dfrac{1}{n}\right)\\
&= \dfrac{m\Sigma}{a_nm-1}\left[\left(a_n-1\right)+1-\dfrac{1}{m}\right] + \dfrac{m\Gamma}{a_nm-1}\left(\dfrac{a_n-1}{b_n} + \dfrac{1}{n}\left(1-\dfrac{1}{m}\right)\right) + o\left(\dfrac{1}{b_n}\right)\\
&= \Sigma + \dfrac{m\left(a_n-1\right)\Gamma}{\left(a_nm-1\right)b_n} + \dfrac{m\Gamma}{n\left(a_nm-1\right)}\left(1-\dfrac{1}{m}\right) + o\left(\dfrac{1}{b_n}\right)\\
&= \Sigma + \dfrac{\Gamma\left(a_n m - 1 + 1 - m\right)}{\left(a_n m - 1\right)b_n} + o\left(\dfrac{1}{b_n}\right)\\
&= \Sigma + \dfrac{\Gamma}{b_n} + o\left(\dfrac{1}{b_n}\right) \; .
\end{align*}

%

\subsubsection{Difference in RBM and ABM bias}
\label{subsec:rbm_vs_abm_bias}
We first establish two preliminary results to assist us in calculating the bias of RBM. Let $R_{1,s} = \Cov(Y_{1,1}, Y_{1,1+s})$. For all $k = 1, \dots, m$, we have

\begin{align*}
    &\mathbb{E}[\hat{\mu}_{k} \hat{\mu}_{k}^T]\\
    &= \mathbb{E} \left[ \dfrac{1}{n^{2}} \left(\sum_{s=1}^{n} Y_{k,s} \right)\left(\sum_{s=1}^{n} Y_{k,s} \right)^{T} \right] \\
    &= \dfrac{1}{n^{2}} \mathbb{E} \left[ \sum_{s=1}^{n}Y_{k,s}Y^{T}_{k,s} + \sum_{s \neq t}Y_{k,s} Y^{T}_{k,t} \right] \\
    &= \dfrac{1}{n^{2}} \mathbb{E} \left[ \sum_{s=1}^{n}Y_{k,s}Y^{T}_{k,s} + \sum_{s < t}Y_{k,s} Y^{T}_{k,t} + \sum_{s > t}Y_{k,s} Y^{T}_{k,t} \right] \\
    &= \dfrac{1}{n^{2}} \left[ n (R_{0} + (\mathbb{E}[Y_{k,1}])\mathbb{E}[Y_{k,1}])^{T}) + \sum_{s=1}^{n-1} (n - s) (R_{s} + \mathbb{E}[Y_{k,1}](\mathbb{E}[Y_{k,1}])^{T})) \right.\\
    & \qquad \qquad \left.+ \sum_{s=1}^{n-1} (n - s) (R_{s}^{T} + \mathbb{E}[Y_{k,1}](\mathbb{E}[Y_{k,1}])^{T})) \right] \\
    &= \dfrac{R_{0}}{n} + \dfrac{1}{n} \sum_{s=1}^{n-1}\left(1 - \dfrac{s}{n}\right) (R_{s} + R_{s}^{T}) + \dfrac{\mathbb{E}[Y_{k,1}](\mathbb{E}[Y_{k,1}])^{T}}{n} + \dfrac{(n-1)}{n}\mathbb{E}[Y_{k,1}](\mathbb{E}[Y_{k,1}])^{T} \\
    &= \dfrac{R_{0}}{n} + \dfrac{1}{n} \sum_{s=1}^{n-1}\left(1 - \dfrac{s}{n}\right) (R_{s} + R_{s}^{T}) + \mathbb{E}[Y_{k,1}](\mathbb{E}[Y_{k,1}])^{T} \, .
\end{align*}

Next, we see that
\begin{align*}
    \mathbb{E} \left[ \hat{\mu} \hat{\mu}^{T} \right] &= \dfrac{1}{m^{2}} \mathbb{E} \left[ \left(\sum_{k=1}^{m}\hat{\mu}_{k} \right)\left(\sum_{k=1}^{m}\hat{\mu}_{k} \right)^{T} \right] \\
    &= \dfrac{1}{m^{2}} \mathbb{E} \left[ \sum_{k=1}^{m} \hat{\mu}_{k}\hat{\mu}^{T}_{k} + \sum_{k \neq l} \hat{\mu}_{k} \hat{\mu}^{T}_{l} \right] \\
    &= \dfrac{1}{m^{2}} \left[m \mathbb{E}[\hat{\mu}_{1} \hat{\mu}_{1}^{T}] + m(m-1) \mathbb{E}[\hat{\mu}_{1}](\mathbb{E}[\hat{\mu}_{1}])^{T} \right] \\
    &= \dfrac{1}{m} \left( \mathbb{E}[\hat{\mu}_{1}\bar{Y}_{1}^{T}] + (m-1) \mathbb{E}[\hat{\mu}_{1}](\mathbb{E}[\hat{\mu}_{1}])^{T} \right)
\end{align*}
Using the above two results, we can calculate the expectation of the RBM estimator.
\begin{align*}
    & \mathbb{E}[\hat{\Sigma}_{R, b_n}] \\
    &= \mathbb{E} \left[\dfrac{b_n}{a_nm-1}\sum_{k=1}^{m} \sum_{l=1}^{a_n} \left(\bar{Y}_{k,l} - \hat{\mu} \right) \left(\bar{Y}_{k,l} - \hat{\mu} \right) ^T \right] \\
    &= \mathbb{E} \left[\dfrac{b_n}{a_nm-1}\sum_{k=1}^{m} \sum_{l=1}^{a_n} \left(\bar{Y}_{k,l}\bar{Y}_{k,l}^T + \hat{\mu} \hat{\mu}^T - \hat{\mu}\bar{Y}_{k,l}^T - \bar{Y}_{k,l}\hat{\mu}^T \right) \right] \\
    &= \dfrac{b_n}{a_nm-1}\mathbb{E}\left[\sum_{k=1}^{m} \sum_{l=1}^{a_n} \bar{Y}_{k,l}\bar{Y}_{k,l}^T - a_n m \hat{\mu} \hat{\mu}^T \right] \\
    &= \dfrac{b_n}{a_n m-1} \left(a_n m \mathbb{E}[\bar{Y}_{1,1}\bar{Y}_{1,1}^T] - a_n m \mathbb{E} \left[\hat{\mu} \hat{\mu}^T \right] \right)  \\
    &= \dfrac{b_n a_n m}{a_n m-1}\left[\dfrac{R_{0}}{b_n} + \dfrac{1}{b_n} \sum_{s=1}^{b_n -1}\left(1 - \dfrac{s}{b_n}\right) (R_{s} + R_{s}^{T}) + \mathbb{E}[Y_{1,1}](\mathbb{E}[Y_{1,1}])^{T} \right. \\
    & \quad \left. - \dfrac{1}{m} \left( \mathbb{E}[\bar{Y}_{1}\bar{Y}_{1}^{T}] + (m-1) \mathbb{E}[\bar{Y}_{1,1}](\mathbb{E}[\bar{Y}_{1,1}])^{T} \right) \right]\\
    &= \dfrac{b_n a_n m}{a_n m-1}\left[ \dfrac{R_{0}}{b_n} + \dfrac{1}{b_n} \sum_{s=1}^{b_n -1}\left(1 - \dfrac{s}{b_n}\right) (R_{s} + R_{s}^{T}) + \mathbb{E}[Y_{1,1}](\mathbb{E}[Y_{1,1}])^{T} \right. \\
    & \quad \left. - \dfrac{1}{m}\left( \dfrac{R_{0}}{n} + \dfrac{1}{n} \sum_{s=1}^{n-1}\left(1 - \dfrac{s}{n}\right) (R_{s} + R_{s}^{T}) + \mathbb{E}[Y_{1,1}](\mathbb{E}[Y_{1,1}])^{T} + (m-1) \mathbb{E}[Y_{1,1}](\mathbb{E}[Y_{1,1}])^{T} \right) \right] \\    
    &= \dfrac{b_n a_n m}{a_n m-1} \left[ R_{0} \left(\dfrac{1}{b_n} - \dfrac{1}{mn}\right) + \sum_{s=1}^{b_n-1}(R_{s} + R_{s}^{T}) \left( \dfrac{1}{b_n} - \dfrac{1}{mn}\right) -  \sum_{s=1}^{b_n-1}t (R_{s}+R_{s}^{T}) \left( \dfrac{1}{b_n^{2}} - \dfrac{1}{m n^{2}}\right) \right. \\
    & \quad \quad \left. - \dfrac{1}{mn}\sum_{s=b_n}^{n-1}\left(1 - \dfrac{s}{n}\right) (R_{s} + R_{s}^{T}) + \mathbb{E}[Y_{1,1}](\mathbb{E}[Y_{1,1}])^{T} - \dfrac{1}{m}m\mathbb{E}[Y_{1,1}](\mathbb{E}[Y_{1,1}])^{T} \right] \\
    &= \dfrac{b_n a_n m}{a_n m-1} \left[ R_{0} \left(\dfrac{1}{b_n} - \dfrac{1}{mn}\right) + \sum_{s=1}^{b_n-1}(R_{s}+R_{s}^{T}) \left( \dfrac{1}{b_n} - \dfrac{1}{mn}\right) -  \sum_{s=1}^{b_n-1}s (R_{s}+R_{s}^{T})\left( \dfrac{1}{b_n^{2}} - \dfrac{1}{m n^{2}}\right)  \right. \\
    & \quad \quad \left. - \dfrac{1}{mn}\sum_{s=b_n}^{n-1}\left(1 - \dfrac{s}{n}\right) (R_{s} + R_{s}^{T}) \right] \\
    &= R_{0} + \sum_{s=1}^{b_n-1} (R_{s}+R_{s}^{T}) - \dfrac{(m a^{2}_n - 1)}{a_n b_n (a_n m-1)} \sum_{s=1}^{b_n-1} s (R_{s}+R_{s}^{T}) - \dfrac{1}{a_n m - 1} \sum_{s=b_n}^{n-1} \left(1 - \dfrac{s}{n}\right) (R_{s}+R_{s}^{T}) 
\end{align*}
Since ABM takes a simple average of the BM estimators from the $m$ chains, the bias of ABM is the same as the bias of a single chain BM estimator. Setting $m=1$ in RBM is equivalent to the single chain BM estimator, 
\begin{align*}
&\mathbb{E}[\hat{\Sigma}_{A, b_n}] \\
&= \mathbb{E}[\hat{\Sigma}_{R, b_n}] |_{m=1} \\
&= R_{0} + \sum_{s=1}^{b_n-1} (R_{s}+R_{s}^{T}) - \dfrac{(a_n^{2} - 1)}{a_n b_n (a_n-1)} \sum_{s=1}^{b_n-1} s (R_{s}+R_{s}^{T}) - \dfrac{1}{a_n - 1} \sum_{s=b_n}^{n-1} \left(1 - \dfrac{s}{n}\right) (R_{s} + R_{s}^{T}) \\
&= R_{0} + \sum_{s=1}^{b_n-1} (R_{s}+R_{s}^{T}) - \dfrac{(a_n + 1)}{a_n b_n} \sum_{s=1}^{b_n-1} s (R_{s}+R_{s}^{T}) - \dfrac{1}{a_n - 1} \sum_{s=b_n}^{n-1} \left(1 - \dfrac{s}{n}\right) (R_{s}+R_{s}^{T})
\end{align*}
Hence, we can write the difference in the bias of RBM and ABM as 
\begin{align*}
    & \mathbb{E}[\hat{\Sigma}_{R, b_n}] - \mathbb{E}[\hat{\Sigma}_{A, b_n}] \\
    &= R_{0} + \sum_{s=1}^{b_n-1} (R_{s} + R_{s}^{T})  - \dfrac{(m a_n^{2} - 1)}{a_n b_n (a_n m-1)} \sum_{s=1}^{b_n-1} s (R_{s} + R_{s}^{T}) - \dfrac{1}{a_n m - 1} \sum_{s=b_n}^{n-1} \left(1 - \dfrac{s}{n}\right) (R_{s}+R_{s}^{T})  \\
    & \quad - \left( R_{0} + \sum_{s=1}^{b_n-1} (R_{s} + R_{s}^{T}) - \dfrac{(a_n + 1)}{a_nb_n} \sum_{s=1}^{b_n-1} s (R_{s} + R_{s}^{T})  - \dfrac{1}{a_n - 1} \sum_{s=b_n}^{n-1} \left(1 - \dfrac{s}{n}\right) (R_{s} + R_{s}^{T}) \right) \\
    &= R_{0} + \sum_{s=1}^{b_n-1} (R_{s} + R_{s}^{T}) - \left(R_{0} + \sum_{s=1}^{b_n-1} (R_{s} + R_{s}^{T}) \right) + \dfrac{(a_n+1)}{a_n b_n} \sum_{s=1}^{b_n-1} s (R_{s} + R_{s}^{T})  \\
    & \quad  - \dfrac{(ma^{2}_n - 1)}{a_n b_n (a_n m-1)} \sum_{s=1}^{b_n-1} s (R_{s} + R_{s}^{T}) + \dfrac{1}{a_n-1} \sum_{s=b_n}^{n-1} \left(1 - \dfrac{s}{n}\right) (R_{s} + R_{s}^{T}) \\
    & \quad - \dfrac{1}{a_n m - 1} \sum_{s=b_n}^{n-1} \left(1 - \dfrac{s}{n}\right) (R_{s} + R_{s}^{T}) \\
    &= \dfrac{1}{a_n b_n} \sum_{s=1}^{b_n-1}s (R_{s} + R_{s}^{T}) \left( (a_n+1) - \dfrac{ma^{2}_n -1 }{a_n m-1}\right)  + \sum_{s=b_n}^{n-1} \left(1 - \dfrac{s}{n}\right) (R_{s} + R_{s}^{T}) \left( \dfrac{1}{a_n-1} - \dfrac{1}{a_nm-1} \right) \\
    &= \dfrac{1}{a_nb_n} \sum_{s=1}^{b_n-1}s (R_{s} + R_{s}^{T}) \left( \dfrac{ma^{2}_n -a_n + a_n m - 1 -  ma^{2}_n + 1}{a_nm-1} \right)  \\
    & \quad + \sum_{s=b_n}^{n-1} \left(1 - \dfrac{s}{n}\right) (R_{s} + R_{s}^{T}) \left( \dfrac{a_n m-1 - a_n + 1}{(a_n-1)(a_nm-1)} \right) \\
    &= \sum_{s=1}^{b_n-1} \dfrac{s}{n} (R_{s} + R_{s}^{T}) \dfrac{a_n(m-1)}{a_nm-1} + \sum_{s=b_n}^{n-1} \left(1 - \dfrac{s}{n}\right) (R_{s} + R_{s}^{T}) \dfrac{a_n(m-1)}{(a_n-1)(a_nm-1)}\\
    &= \dfrac{a_n(m-1)}{a_nm-1} \left[\sum_{s=1}^{b_n-1} \dfrac{s}{n}(R_{s} + R_{s}^{T}) + \dfrac{1}{a_n-1}\sum_{s=b_n}^{n-1}\left(1 - \dfrac{s}{n}\right) (R_{s} + R_{s}^{T}) \right] 
\end{align*}
The diagonal terms of the difference in the bias of RBM and ABM are
\begin{align*}
  &\mathbb{E}[\hat{\Sigma}_{R, b_n}^{kk}] - \mathbb{E}[\hat{\Sigma}_{A, b_n}^{kk}] \\
  &= \dfrac{a_n(m-1)}{a_nm-1} \left[\sum_{s=1}^{b_n-1} \dfrac{s}{n}(R_{s}^{kk} + (R^{T})_{s}^{kk}) + \dfrac{1}{a_n-1}\sum_{s=b_n}^{n-1}\left(1 - \dfrac{s}{n}\right) (R_{s}^{kk} + (R^{T})_{s}^{kk}) \right] \\
  &= \dfrac{2a_n(m-1)}{a_n m-1} \left[\sum_{s=1}^{b_n-1} \dfrac{s}{n}(R_{s}^{kk}) + \dfrac{1}{a_n-1}\sum_{s=b}^{n-1}\left(1 - \dfrac{s}{n}\right) (R_{s}^{kk}) \right] 
\end{align*}
\end{proof}

\subsection{Variance of RBM} 
\label{subsec:variance_RBM}

Recall that $\hat{\Sigma}_{R,L}$ is the lugsail RBM estimator and $\hat{\Sigma}_{A,L}$ is the lugsail averaged batch means estimator. Further, let $\hat{\Sigma}^{ij}_{R,L}$ and $\hat{\Sigma}^{ij}_{A,L}$ be the $\left(i,j\right)$th element of the matrices, respectively. We will prove that the variance of both estimators is equivalent for large sample sizes. Due to the strong consistency proof, as $n \to \infty$,
\begin{equation}
\label{eq:rbm_abm_consis}
\left\|\hat{\Sigma}_{R,L} - \dfrac{m\left(a_n-1\right)}{a_nm - 1} \hat{\Sigma}_{A,L} \right\| \to 0 \text{ with probability 1}\,. 
\end{equation}
For ease of notation, set $C = m\left(a_n-1\right)/\left(a_n m - 1\right)$. Further, define
\begin{align*}
g_1\left(n\right) &= 2\left(\dfrac{1}{1-c} \dfrac{a_nm}{a_nm - 1} +  \dfrac{c}{1-c} \dfrac{a_nm}{a_nrm - 1} \right)\dfrac{b_n \psi\left(n\right)^2}{n^2} \\ 
g_2\left(n\right) &= 2\left(\dfrac{1}{1-c} \dfrac{a_nm}{a_nm - 1} +  \dfrac{c}{1-c} \dfrac{a_nm}{a_nrm - 1} \right) \left(\Sigma_{ii}+\Sigma_{jj}\right)\left(1+\epsilon\right)\dfrac{b_n\psi\left(n\right) \sqrt{2n \log\log n}}{n^2}\\
g_3\left(n\right) &= 2\left(\dfrac{1}{1-c} \dfrac{a_nm}{a_nm - 1} +  \dfrac{c}{1-c} \dfrac{a_nm}{a_nrm - 1} \right)\Sigma_{ii}\Sigma_{jj}\left(1 + \epsilon\right)^2 \dfrac{2b_n\log\log n}{n}\,.
\end{align*}
Since $\psi(n) = O(\sqrt{n \log \log n})$, if $\left(b_n\log \log n\right)/n \to 0$, then $g_1, g_2, g_3 \to 0$.
From the steps in Theorem~\ref{thm:RBM_consistency},
\begin{align*}
 &\left| \hat{\Sigma}_{R,L}^{ij} - C \hat{\Sigma}_{A,L}^{ij} \right|\\
 & =  \left| \dfrac{1}{1-c} \left(\hat{\Sigma}_{R,b_n}^{ij} - C \hat{\Sigma}_{A,b_n}^{ij} \right) - \dfrac{c}{1-c} \left(\hat{\Sigma}_{R,b_n/r}^{ij} - C \hat{\Sigma}_{A,b_n/r}^{ij} \right)  \right|\\ 
 & = \left| \dfrac{1}{1-c} \left( \dfrac{a_nb_n}{a_nm - 1} \sum_{k=1}^{m} \left(\hat{\mu}_{k}^{(i)} - \hat{\mu}^{(i)}\right)\left(\hat{\mu}_{k}^{(j)} - \hat{\mu}^{(j)}\right)^T \right) - \dfrac{c}{1-c} \left( \dfrac{a_nb_n}{a_nrm - 1} \sum_{k=1}^{m} \left(\hat{\mu}_{k}^{(i)} - \hat{\mu}^{(i)}\right)\left(\hat{\mu}_{k}^{(j)} - \hat{\mu}^{(j)}\right)^T  \right) \right| \\
 &\leq \dfrac{1}{1-c} \dfrac{a_nb_n}{a_nm - 1}\left| \sum_{k=1}^{m} \left(\hat{\mu}_{k}^{(i)} - \hat{\mu}^{(i)}\right)\left(\hat{\mu}_{k}^{(j)} - \hat{\mu}^{(j)}\right)^T \right| + \dfrac{c}{1-c} \dfrac{a_nb_n}{a_nrm - 1}\left| \sum_{k=1}^{m} \left(\hat{\mu}_{k}^{(i)} - \hat{\mu}^{(i)}\right)\left(\hat{\mu}_{k}^{(j)} - \hat{\mu}^{(j)}\right)^T \right| \\ 
 & =  \left(\dfrac{1}{1-c} \dfrac{a_nb_n}{a_nm - 1} +  \dfrac{c}{1-c} \dfrac{a_nb_n}{a_nrm - 1} \right)\left| \sum_{k=1}^{m} \left(\hat{\mu}_{k}^{(i)} - \hat{\mu}^{(i)}\right)\left(\hat{\mu}_{k}^{(j)} - \hat{\mu}^{(j)}\right)^T \right|\\
 & \leq D^2g_1\left(n\right) + Dg_2\left(n\right) + g_3\left(n\right)\,.
\end{align*}
By \eqref{eq:rbm_abm_consis}, there exists an $N_0$ such that
\begin{align*}
\left(\hat{\Sigma}_{R,L}^{ij} - C \hat{\Sigma}_{A,L}^{ij} \right)^2 &= \left(\hat{\Sigma}_{R,L}^{ij} - C \hat{\Sigma}_{A,L}^{ij} \right)^2 \, I\left(0 \leq n \leq N_0\right) + \left(\hat{\Sigma}_{R,L}^{ij} - C \hat{\Sigma}_{A,L}^{ij} \right)^2 \, I\left(n > N_0\right)\\
& \leq \left(\hat{\Sigma}_{R,L}^{ij} - C \hat{\Sigma}_{A,L}^{ij} \right)^2 \, I\left(0 \leq n \leq N_0\right) +  \left(D^2g_1\left(n\right) + Dg_2\left(n\right) + g_3\left(n\right) \right)^2 I\left(n > N_0\right)\\
& := g_n^*\left(X_{11}, \dots, X_{1n}, \dots, X_{m1}, \dots, X_{mn}\right)\,.
\end{align*}
But since by assumption $\E D^4 <\infty$, $\E \|h\|^4 <\infty$,
\[
\E \left| g_n^* \right| \leq  \E \left[\left(\hat{\Sigma}_{R,L}^{ij} - C \hat{\Sigma}_{A,L}^{ij} \right)^2 \right] + \E \left[\left(D^2g_1\left(n\right) + Dg_2\left(n\right) + g_3\left(n\right) \right)^2 \right] < \infty\,.
\]
Thus, $\E \left| g_n^* \right| < \infty$ and further as $n \to \infty$, $g_n \to 0$ under the assumptions. Since $g_1, g_2, g_3 \to 0$, $\E g_n^* \to 0$. By the majorized convergence theorem \citep{zeid:2013}, as $n \to \infty$,
\begin{equation}
\label{eq:squared_mean_diff}
  \E \left[\left(\hat{\Sigma}_{R,L}^{ij} - C \hat{\Sigma}_{A,L}^{ij} \right)^2 \right] \to 0\,.
\end{equation}

We will use \eqref{eq:squared_mean_diff} to show that the variances are equivalent. Define,
\[
\Omega\left(\hat{\Sigma}_{R,L}^{ij}, \hat{\Sigma}_{A,L}^{ij} \right) = \Var\left(\hat{\Sigma}_{R,L}^{ij} - C \hat{\Sigma}_{A,L}^{ij} \right) + 2 C \E\left[ \left(\hat{\Sigma}_{R,L}^{ij} - C \hat{\Sigma}_{A,L}^{ij} \right) \left(\hat{\Sigma}_{A,L}^{ij}  - \E \left( \hat{\Sigma}_{A,L}^{ij} \right) \right) \right] \; .
\]
We will show that the above is $o\left(1\right)$. Using Cauchy-Schwarz inequality followed by \eqref{eq:squared_mean_diff},
\begin{align*}
& \left|  \Omega\left(\hat{\Sigma}_{R,L}^{ij}, \hat{\Sigma}_{A,L}^{ij} \right) \right|\\
& \leq \left| \Var\left(\hat{\Sigma}_{R,L}^{ij} - C \hat{\Sigma}_{A,L}^{ij} \right) \right| + \left| 2 C \E\left[ \left(\hat{\Sigma}_{R,L}^{ij} - C \hat{\Sigma}_{A,L}^{ij} \right) \left(\hat{\Sigma}_{A,L}^{ij}  - \E \left( \hat{\Sigma}_{A,L}^{ij} \right) \right) \right]\right| \\ 
& = \E\left[\left(\hat{\Sigma}_{R,L}^{ij} - C \hat{\Sigma}_{A,L}^{ij} \right)^2 \right] + 2 C \left| \left(\E\left[ \left(\hat{\Sigma}_{R,L}^{ij} - C \hat{\Sigma}_{A,L}^{ij} \right)^2 \right]  \Var\left(\hat{\Sigma}_{A,L}^{ij}  \right)   \right)^{1/2}\right| \\ 
& = o\left(1\right) + 2C\left(o\left(1\right) \left(O\left( \dfrac{b_n}{n}\right)  + o\left( \dfrac{b_n}{n}\right) \right)  \right) \\ 
& = o\left(1\right)\,.
\end{align*}
Finally,
\begin{align*}
 & \Var\left(\hat{\Sigma}_{R,L}^{ij} \right)\\
  & = \E \left[ \left(\hat{\Sigma}_{R,L}^{ij}  - \E \left[\hat{\Sigma}_{R,L}^{ij}  \right] \right)^2 \right]\\
& = \E \left[ \left(\hat{\Sigma}_{R,L}^{ij} \pm C\hat{\Sigma}_{A,L}^{ij} \pm C\E \left[ \hat{\Sigma}_{A,L}^{ij}\right] - \E \left[\hat{\Sigma}_{R,L}^{ij}  \right] \right)^2 \right]\\
& = \E\left[ \left( \left(\hat{\Sigma}_{R,L}^{ij} - C\hat{\Sigma}_{A,L}^{ij} \right) + C\left(\hat{\Sigma}_{A,L}^{ij}  - \E\left[\hat{\Sigma}_{A,L}^{ij}\right]\right) + \left(C\E\left[\hat{\Sigma}_{A,L}^{ij}\right] - \E \left[\hat{\Sigma}_{R,L}^{ij}  \right] \right)  \right)^2 \right] \\ 
& = C^{2} \E\left[ \left(\hat{\Sigma}_{A,L}^{ij}  - \E\left[\hat{\Sigma}_{A,L}^{ij}\right]\right)^2 \right] + \E \left[ \left(\left(\hat{\Sigma}_{R,L}^{ij} - C\hat{\Sigma}_{A,L}^{ij} \right) + \left(C\E\left[\hat{\Sigma}_{A,L}^{ij}\right] - \E \left[\hat{\Sigma}_{R,L}^{ij}  \right] \right) \right)^2 \right] \\
& \quad \quad + 2C\E\left[\left(\hat{\Sigma}_{A,L}^{ij}  - \E\left[\hat{\Sigma}_{A,L}^{ij}\right]\right) \left(\hat{\Sigma}_{R,L}^{ij} - C\hat{\Sigma}_{A,L}^{ij} \right)\right] \\
& \quad \quad + 2C \E\left[\left(\hat{\Sigma}_{A,L}^{ij}  - \E\left[\hat{\Sigma}_{A,L}^{ij}\right]\right) \left(C\E\left[\hat{\Sigma}_{A,L}^{ij}\right] - \E \left[\hat{\Sigma}_{R,L}^{ij}  \right] \right)\right]\\
& = C^{2}\Var\left( \hat{\Sigma}_{A,L}^{ij}\right) + \Var\left(\hat{\Sigma}_{R,L}^{ij} - C \hat{\Sigma}_{A,L}^{ij} \right) + 2 C \E\left[ \left(\hat{\Sigma}_{R,L}^{ij} - C \hat{\Sigma}_{A,L}^{ij} \right) \left(\hat{\Sigma}_{A,L}^{ij}  - \E \left( \hat{\Sigma}_{A,L}^{ij} \right) \right) \right] + o\left(1\right)\\
& = C^{2}\Var\left( \hat{\Sigma}_{A,L}^{ij}\right) + \Omega\left(\hat{\Sigma}_{R,L}^{ij}, \hat{\Sigma}_{A,L}^{ij} \right) + o\left(1\right)\\
& = C^{2}\Var\left( \hat{\Sigma}_{A,L}^{ij}\right) + o\left(1\right) + o\left(1\right)\,.
\end{align*}
Thus,
\begin{align*}
 &\Var\left(\hat{\Sigma}_{R,L}^{ij} \right) \\
 & = \dfrac{m^{2}\left(a_n-1\right)^{2}}{\left(a_n m - 1\right)^{2}}\Var\left(\hat{\Sigma}_{A,L}^{ij} \right) +o\left(1\right)\\
&= \dfrac{m\left(a_n-1\right)^{2}}{\left(a_n m - 1\right)^{2}}\Var\left(\hat{\Sigma}_{B,l}^{ij} \right) +o\left(1\right)\\ 
& = \dfrac{m\left(a_n-1\right)^{2}}{\left(a_n m - 1\right)^{2}}\left[\dfrac{b_n}{n}\left(\dfrac{1}{r} + \dfrac{r - 1}{r\left(1 - c\right)^{2}}\right)\left(\Sigma^{2}_{ij} + \Sigma_{ii}\Sigma_{jj}\right) + o\left(\dfrac{b_n}{n}\right)\right] +o\left(1\right)\\
&= \dfrac{m\left(a_n-1\right)^{2}}{a_n\left(a_n m - 1\right)^{2}}\left[\left(\dfrac{1}{r} + \dfrac{r - 1}{r\left(1 - c\right)^{2}}\right)\left(\Sigma^{2}_{ij} + \Sigma_{ii}\Sigma_{jj}\right)\right] + o\left(1\right)\\
&= \dfrac{m^{2}\left(a_n-1\right)^{2}}{a_nm\left(a_n m - 1\right)^{2}}\left[\left(\dfrac{1}{r} + \dfrac{r - 1}{r\left(1 - c\right)^{2}}\right)\left(\Sigma^{2}_{ij} + \Sigma_{ii}\Sigma_{jj}\right)\right] + o\left(1\right)\\
&= \left(\dfrac{\left(a_n m\right)^{2} + m^{2} - 2a_n m^{2}}{a_nm\left(a_n m - 1\right)^{2}}\right)\left[\left(\dfrac{1}{r} + \dfrac{r - 1}{r\left(1 - c\right)^{2}}\right)\left(\Sigma^{2}_{ij} + \Sigma_{ii}\Sigma_{jj}\right)\right] + o\left(1\right)\\
&= \left(\dfrac{\left(\left(a_n m\right)^{2} - 2a_n m + 1 + 2a_n m - 1 + m^{2} - 2a_n m^{2}\right)}{a_nm\left(a_n m - 1\right)^{2}}\right)\left[\left(\dfrac{1}{r} + \dfrac{r - 1}{r\left(1 - c\right)^{2}}\right)\left(\Sigma^{2}_{ij} + \Sigma_{ii}\Sigma_{jj}\right)\right] + o\left(1\right)\\
&= \left(\dfrac{1}{a_nm} - \dfrac{2\left(m - 1\right)}{\left(a_n m - 1\right)^{2}} + \dfrac{m^{2} - 1}{a_nm\left(a_n m - 1\right)^{2}}\right)\left[\left(\dfrac{1}{r} + \dfrac{r - 1}{r\left(1 - c\right)^{2}}\right)\left(\Sigma^{2}_{ij} + \Sigma_{ii}\Sigma_{jj}\right)\right] + o\left(1\right)\\
&= \dfrac{1}{a_nm}\left[\left(\dfrac{1}{r} + \dfrac{r - 1}{r\left(1 - c\right)^{2}}\right)\left(\Sigma^{2}_{ij} + \Sigma_{ii}\Sigma_{jj}\right)\right] + o\left(1\right)\\
&= \dfrac{b_n}{mn}\left[\left(\dfrac{1}{r} + \dfrac{r - 1}{r\left(1 - c\right)^{2}}\right)\left(\Sigma^{2}_{ij} + \Sigma_{ii}\Sigma_{jj}\right)\right] + o\left(1\right) \; .
 \end{align*}

\subsection{Proof of Theorem~\ref{thm:naive_convergence}}
\label{subsec: naive_convergence}
\begin{proof}[Proof of Theorem~\ref{thm:naive_convergence}]

The proof follows from the fact that under Assumption~\ref{ass:mixing}, a CLT holds for both $\hat{\mu}_k$ and $\hat{\mu}$. That is,
 \[
\sqrt{n}(\hat{\mu}_k - \mu)\overset{d}{\to} N(0, \Sigma)\, \quad \text{ and } \quad \sqrt{mn}(\hat{\mu} - \mu)\overset{d}{\to} N(0, \Sigma)\,.
 \]
By Lemma~\ref{lem:Y_diff_simpli},
\begin{align*}
 \hat{\Sigma}_N &= \dfrac{1}{m-1} \left[n\ds \sum_{k=1}^{m} (\hat{\mu}_k - \mu)(\hat{\mu}_k - \mu)^T - mn (\hat{\mu} - \mu)(\hat{\mu} - \mu)^T \right]\\	& = \dfrac{1}{m-1} \left[\ds \sum_{k=1}^{m} \sqrt{n}(\hat{\mu}_k - \mu) \sqrt{n}(\hat{\mu}_k - \mu)^T - \sqrt{mn} (\hat{\mu} - \mu)\sqrt{mn}(\hat{\mu} - \mu)^T \right]\,.
 \end{align*} 
By a standard argument of asymptotic independence of $\hat{\Sigma}_N$ and $\hat{\mu}$,
\[
\hat{\Sigma}_N \overset{d}{\to}W_p(\Sigma, m-1)/(m-1)\,.
\]
\end{proof}
 
\subsection{Bivariate Normal Gibbs Asymptotic Variance}
\label{subsec:Gibbs_var}

The diagonals of $\Sigma$ are indirectly obtained by \cite{gey:1995} using a slightly different proof technique. For $t \geq 0$, let $\epsilon_{1,t} \sim N\left(0, \omega_1 - \rho^{2}/{\omega_2}\right)$ and let $\epsilon_{2,t} \sim N(0, \omega_2 - \rho^2/\omega_1)$, independent of each other. For a given state at time $t-1$, the next state is drawn from
\begin{equation}
\label{eqn:gibbs1_t}
X_{1,t}  =  \mu_{1} + \dfrac{\rho}{\omega_2}\left(X_{2,t-1} - \mu_{2}\right) + \epsilon_{1,t}\,
\end{equation}
\begin{equation}
\label{eqn:gibbs2_t}
X_{2,t}  =  \mu_{2} + \dfrac{\rho}{\omega_1}\left(X_{1,t} - \mu_{1}\right) +  \epsilon_{2,t}
\end{equation}
Additionally, for $t \geq 1$, define, 
\[
w_t = \dfrac{\rho}{\omega_2} \epsilon_{2,t-1} + \epsilon_{1,t} \sim N\left(0, \omega_1\left(1 - \dfrac{\rho^{4}}{\omega_1^{2} \omega_2^{2}}\right)\right)\,.
\]
Substituting $X_{2,t - 1}$ from \eqref{eqn:gibbs2_t} to \eqref{eqn:gibbs1_t} gives,
\begin{align*}
    X_{1,t} &= \mu_{1} + \dfrac{\rho}{\omega_2}\left(\dfrac{\rho}{\omega_1}X_{1,t-1} - \dfrac{\rho}{\omega_1}\mu_{1} + \epsilon_{2,t-1}\right) + \epsilon_{1,t}\\
    &= \mu_{1} + \dfrac{\rho^{2}}{\omega_1 \omega_2}X_{1,t-1} - \dfrac{\rho^{2}}{\omega_1\omega_2}\mu_{1} + \dfrac{\rho}{\omega_2} \epsilon_{2,t-1} + \epsilon_{1,t}\\
    &= \mu_{1}\left(1 - \dfrac{\rho^{2}}{\omega_1\omega_2}\right) + \dfrac{\rho^{2}}{\omega_1\omega_2}X_{1,t-1} + w_t\,. \numberthis \label{eqn:gibbs_AR1}
\end{align*}
This is an autoregressive process of order 1 (AR(1) process). The process is also stationary by the assumption that $\rho^2 < \omega_1\omega_2$ and the stationary distribution is $N(\mu_1, \omega_1)$. Thus $\Var(X_1) = \omega_1$ and similarly $\Var(X_2) = \omega_2$.  Using the properties of an AR(1) process, 
\[
\Cov\left(X_{1,1}, X_{1,1+k}\right) = \Var\left(X_{1}\right)\left(\dfrac{\rho^{2} }{\omega_1\omega_2}\right)^{k} = \omega_1\left( \dfrac{\rho^{2}}{\omega_1\omega_2}\right)^{k}\,.
\]

The diagonal terms of the asymptotic covariance matrix for $i = 1,2$ are
\[
\Sigma_{ii} = \lim_{n\to\infty} n \Var\left(\dfrac{1}{n}\sum_{t=1}^{n} X_{i,t}\right)\,.
\]
Without loss of generality, consider $i = 1$
\begin{align*}
  \Sigma_{11} &= \Var\left(X_{1}\right) + 2\sum_{k=1}^{\infty} \Cov\left(X_{1,1}, X_{1,1+k}\right)\\
    &= \omega_1 + 2\sum_{j=2}^{\infty} \omega_1 \left(\dfrac{\rho^{2}}{\omega_1 \omega_2}\right)^{k}
    = \omega_1\left(\dfrac{\omega_1\omega_2 + \rho^{2}}{\omega_1\omega_2 - \rho^{2}}\right) \; .
\end{align*}
Similarly, 
\[
\Sigma_{22} = \omega_2 \left(\dfrac{\omega_1 \omega_2 + \rho^{2}}{\omega_1\omega_2 - \rho^{2}}\right) \; .
\]
Consider the off-diagonal terms,
\begin{equation}
\label{eq:binorm_cov_limit}
  \Sigma_{12} = \Sigma_{21} = \Cov(X_{1,1}, X_{2,1}) + \sum_{k=1}^{\infty} \left[\Cov\left( X_{1,1} , X_{2,1 + k}\right) + \Cov\left( X_{2,1} , X_{1,1 + k}\right)  \right]\,.
\end{equation}
We have $\Cov\left( X_{1,1} , X_{2,1 + k}\right) = \mathbb{E}(X_{1,1} X_{2,1 + k}) - \mu_1 \mu_2$ and using \eqref{eqn:gibbs2},
\begin{align*}
    \mathbb{E}\left(X_{1,1}X_{2,1+k}\right) &= \mathbb{E}\left[X_{1,1}\left(\mu_{2} - \dfrac{\rho}{\omega_1 }\mu_{1} + \dfrac{\rho}{\omega_1 } X_{1,1+k} + \epsilon_{2,1+k}\right)\right]\\
    &= \mu_{1}\mu_{2} - \dfrac{\rho}{\omega_1}\mu_{1}^{2} + \dfrac{\rho}{\omega_1}\mathbb{E}\left(X_{1,1}X_{1,1+k}\right) + 0\\
    &= \mu_{1}\mu_{2} - \dfrac{\rho}{\omega_1}\mu_{1}^{2} + \dfrac{\rho}{\omega_1}\left[\Cov\left(X_{1,1}, X_{1,1+k}\right) + \mathbb{E}\left(X_{1,1}\right)\mathbb{E}\left(X_{1,1+k}\right)\right]\\
    &= \mu_{1}\mu_{2} - \dfrac{\rho}{\omega_1}\mu_{1}^{2} + \dfrac{\rho}{\omega_1}\left(\omega_1 \left(\dfrac{\rho^{2}}{\omega_1 \omega_2}\right)^{k} + \mu_{1}^{2}\right)\\
    &= \mu_{1}\mu_{2} + \rho\left(\dfrac{\rho^{2}}{\omega_1\omega_2}\right)^{k}\,. \numberthis \label{eq:binorm_cov_term_1}
\end{align*}
Similarly, simplification of $\mathbb{E}\left(X_{1,1}X_{2,1+k}\right)$ using \eqref{eqn:gibbs1} gives
\begin{align*}
    \mathbb{E}\left(X_{2,1}X_{1,1+k}\right) &= \mathbb{E}\left[X_{2,1}\left(\mu_{1} - \dfrac{\rho}{\omega_2}\mu_{2} + \dfrac{\rho}{\omega_2}X_{2,k} + \epsilon_{1,k} \right)\right]\\
    &= \mu_{1}\mu_{2} - \dfrac{\rho}{\omega_2}\mu_{2}^{2} + \dfrac{\rho}{\omega_2}\mathbb{E}\left(X_{2,1}X_{2,k}\right) + 0\\
    &= \mu_{1}\mu_{2} - \dfrac{\rho}{\omega_2}\mu_{2}^{2} + \dfrac{\rho}{\omega_2}\left( \Cov\left(X_{2,1}, X_{2,k}\right) + \mathbb{E}\left(X_{2,1}\right)\mathbb{E}\left(X_{2,k}\right)\right)\\
    &= \mu_{1}\mu_{2} - \dfrac{\rho}{\omega_2}\mu_{2}^{2} + \dfrac{\rho}{\omega_2}\left(\omega_2\left(\dfrac{\rho^{2}}{\omega_1\omega_2}\right)^{k - 1} + \mu_{1}^{2}\right)\\
    &= \mu_{1}\mu_{2} + \rho\left(\dfrac{\rho^{2}}{\omega_1\omega_2}\right)^{k - 1}\,. \numberthis \label{eq:binorm_cov_term_2}
\end{align*}
Using \eqref{eq:binorm_cov_term_1} and \eqref{eq:binorm_cov_term_2} in \eqref{eq:binorm_cov_limit}
\begin{align*}
\Sigma_{12} &=\Sigma_{21} = \Cov(X_{1,1}, X_{2,1}) + \sum_{k=1}^{\infty} \left[\Cov\left( X_{1,1} , X_{2,1 + k}\right) + \Cov\left( X_{2,1} , X_{1,1 + k}\right)  \right]\\
& = \rho + \ds \sum_{k=1}^{\infty} \left[\rho\left(\dfrac{\rho^{2}}{\omega_1\omega_2}\right)^{k} + \rho\left(\dfrac{\rho^{2}}{\omega_1\omega_2}\right)^{k-1}  \right] = \dfrac{2\omega_1\omega_2 \rho}{\omega_1\omega_2 - \rho^2}\,.
\end{align*}

\singlespacing
\bibliographystyle{apalike}
\bibliography{mcref}

\end{document}